\documentclass[11pt]{article}
\usepackage{amssymb,latexsym,amsmath,amsthm,graphicx,subfigure}

\makeatletter \@addtoreset{figure}{section}
\def\thefigure{\thesection.\@arabic\c@figure}
\def\fps@figure{h, t}
\@addtoreset{table}{bsection}
\def\thetable{\thesection.\@arabic\c@table}
\def\fps@table{h, t}
\@addtoreset{equation}{section}

\newcommand\be{{\mathbf e}}

\newcommand\cM{{\mathcal M}}

\newcommand\field{\mathbb}

\newcommand\R{\field{R}}

\newcommand\C{\field{C}}

\newcommand\N{\field{N}}

\newcommand\diag{\operatorname{diag}}

\newcommand\id{\operatorname{\mathrm{Id}}}

\newcommand\rmd{\mathrm{d}}

\newcommand\rmi{\mathrm{i}\mspace{1mu}}
\newcommand\rme{\mathrm{e}}
\newcommand\pairing[2]{\langle {#1}, {#2}\rangle}
\newcommand\mtext[1]{\quad\text{#1}\quad}

\newcommand\Res{\operatorname{Res}}
\newtheorem{theorem}{Theorem}
\newtheorem{proposition}[theorem]{Proposition}
\newtheorem{lemma}[theorem]{Lemma}
\newfont{\tenbi}{cmbxti10}
\newcommand{\la}{\lambda}

\title{The generalized Euler-Poinsot rigid body equations: explicit elliptic 
solutions\footnote{{\small AMS Subject Classification 70E40, 70H06, 37J35,  33E05}} }
\author{Yuri N. Fedorov$^1$, 
  Andrzej J.~Maciejewski$^2$, and Maria Przybylska$^3$,   \\[1em]
   {}$^1$ Department de Matem\'atica Aplicada I, \\
Universitat Politecnica de Catalunya, \\
Barcelona, E-08028 Spain \\
e-mail: Yuri.Fedorov@upc.edu \\[1em]
  {}$^2$Kepler Institute of Astronomy,\\ University of Zielona G\'ora, 
  Licealna 9,  \\
  PL-65-407,  Zielona G\'ora, Poland, \\
  e-mail: maciejka@astro.ia.uz.zgora.pl \\[1em]
  {}$^{3}$Institute of Physics, \\ University of Zielona G\'ora, 
  Licealna 9,  \\
  PL-65-407,  Zielona G\'ora, Poland \\
  e-mail: M.Przybylska@proton.if.uz.zgora.pl}

\begin{document}

\maketitle

\begin{abstract} The classical Euler--Poinsot case of the rigid body
  dynamics admits a class of simple but non-trivial integrable
  generalizations, which modify the Poisson equations describing the
  motion of the body in space.  These generalizations possess first 
  integrals which are polynomial in the angular momenta.

  We consider the modified Poisson equations as a system of linear equations with elliptic coefficients 
and show that all the solutions of it are
  single-valued. By using the vector generalization of the Picard
  theorem, we derive the solutions explicitly in terms of sigma
  functions of the corresponding elliptic curve. The solutions are accompanied with a
  numerical example.  We also compare the generalized Poisson equations
  with 
  the classical 3rd order Halphen equation.
\end{abstract}

\section{Introduction}

As in \cite{BZ}, we consider the following system
\begin{equation} \label{EP1} \dot J\omega = J\omega \times \omega,
  \quad \dot\gamma = \gamma \times { B} \omega, \qquad \omega,\gamma
  \in {\mathbb R}^3,
\end{equation}
which is a certain limit of the Kirchhoff equations describing the
motion of a rigid body in an ideal fluid.  Here $\omega$ is the
angular velocity of the body, $\gamma$ is the linear momentum;
$3\times 3$ matrices $J$, and $B$ are tensors of adjoint masses.  The
first equation in \eqref{EP1} are just the Euler equations describing
the free motion of the body with the inertia tensor $J$.

In the sequel, $J$ and $ B$ are assumed to be arbitrary diagonal
matrices.  In the special case $B=\id_{3}$, the system becomes the
classical integrable Euler-Poinsot case of the rigid body motion. In
this case vector $\gamma$ is a vertical vector fixed in space.  Then, in the general case,
$\omega$ is solved in terms of elliptic functions, and the 3
independent solutions for the vector $\gamma$ are elliptic functions
and elliptic functions of the second kind, see e.g., \cite{Jac,Whitt}.

Setting $M=J\omega$, the system \eqref{EP1} can be rewritten in the
form
\begin{equation}
  \label{EP3}
  M = M \times a M, \qquad \dot\gamma =
  \gamma \times b M ,
\end{equation}
where
\begin{equation}
  \label{eq:1}
  a=\diag(a_1, a_2, a_3):= J^{-1}, \qquad  b=\diag(b_1, b_2, b_3):=BJ^{-1}.
\end{equation}
It has three independent polynomial first integrals
\begin{equation}
  \label{EI} H_1:=\langle M ,a M \rangle  ,\qquad
  H_2:=\langle M , M \rangle, \qquad  H_3:=\langle \gamma, \gamma
  \rangle,
\end{equation}
Here and below $\pairing{x}{y}$
denotes the scalar product of vectors $x$, $y\in\R^3$.  As the 
system~\eqref{EP3} is divergence free, according to the Euler--Jacobi theorem, for its integrability one
additional first integral is required.

In \cite{BZ}, the authors applied the Kovalevskaya--Painlev\'e method
to search for integrable cases of the considered system.  It
was shown that if all the solutions of the system \eqref{EP3} are
meromorphic, or single-valued, then
\begin{equation} \label{cond_k} k^2 a_{32} a_{13}a_{21} + b_1^2
  a_{32}+ b_2^2 a_{13} + b_3^2 a_{21} =0, \quad a_{ij}= a_i- a_j,
\end{equation}
where $k$ is an \emph{odd} integer. Geometrically, the above condition
describes a quadric in ${\mathbb R}^3$ with coordinates $(b_1, b_2,b_3)$. 
In particular, one has $b=k a$. In this case the
equations \eqref{EP3} become what can be called the {modified Euler-Poinsot system}
\begin{equation}
  \label{EP2}
  \dot M = M \times a M, \quad \dot\gamma = k \,
  \gamma \times a M.
\end{equation}
As it was shown in \cite{BZ}, if the condition \eqref{cond_k} is
satisfied, then for {\it odd positive} $k$ the system \eqref{EP3} possess an
additional first integral $H_4$, which is algebraically independent
with ~\eqref{EI}. It is linear in $\gamma$, and of
degree $k$ in $M$, and can be written in the following form
\begin{equation}
  H_4 = \langle P(M), \gamma \rangle , \label{I4}
\end{equation}
where the vector $P(M)$ is given by
\begin{gather}
  P(M) = \diag(M_1, M_2, M_3) \Phi_k (M)\, T,    \label{P_mat} \\
  \Phi_{k}(M) :=\left( A_{1}^{-1} K\right)\cdot \left( A_{3}^{-1} K\right)\cdots
\left( A_{k-2}^{-1} K\right). 
\end{gather}
The matrices $K, A_n$ are defined as follows
\begin{gather*}
  K=\diag \left(M_1^2, M_2^2, M_3^2\right), \qquad
  A_n =\begin{pmatrix} -n\, a_{32}   & b_3 & -b_2 \\
    -b_3 &  -n\, a_{13}  & b_1 \\
    b_2 & -b_1 & -n \, a_{21}
  \end{pmatrix}, \quad n\in\N.
\end{gather*}
The constant vector $T\in\R^3$ in formula~\eqref{P_mat} spans the
kernel of the matrix $A_k$.

Notice that
\begin{equation*}
  \det A_n= -n(n^2 a_{32} a_{13}a_{21} + b_1^2 a_{32}+ b_2^2 a_{13} +
  b_3^2 a_{21}).
\end{equation*}

As will be shown in Section 8, the case of negative odd $k$ can be reduced to
the above one. 

In the simplest non-trivial case $b=ka$ with $k=3$, we have
\begin{equation}
  \label{eq:2}
  P(M)=  P^{(3)} :=\left(P^{(3)}_1 ,  P^{(3)}_2, P^{(3)}_3\right)^T
\end{equation}
with
\begin{align*}
  P^{(3)}_1 & =M_1\,[ (a_{1}a_{2}+8a_{1}^{2}-a_{3}a_{2}+a_{3}a_{1}) M_1^2 \\
 & \qquad + (3a_{3}a_{2}-3a_{3}a_{1}+9a_{1}a_{2})M_2^2 + (9a_{3}a_{1}-3a_{1}a_{2}+3a_{3}a_{2})M_3^2], \\
  P^{(3)}_2 & = M_2\, [(-3a_{3}a_{2}+3a_{3}a_{1}+9a_{1}a_{2})M_1^2 \\
   &\qquad +(a_{3}a_{2}-a_{3}a_{1}+8a_{2}^{2}+a_{1}a_{2})M_2^2 +
    (3a_{3}a_{1}-3a_{1}a_{2}+9a_{3}a_{2})M_3^2 ], \\
  P^{(3)}_3 & =M_3\,[(9a_{3}a_{1}+3a_{1}a_{2}-3a_{3}a_{2})M_1^2 \\
& \qquad +(-3a_{3}a_{1}+3a_{1}a_{2}+9a_{3}a_{2})M_2^2 +
    (a_{3}a_{1}+8a_{3}^{2}-a_{1}a_{2}+a_{3}a_{2})M_3^2 ].
\end{align*}

In \cite{BZ} it was also shown that the vector $P(t):=P(M(t))$,
where $M(t)$ is a solution of the Euler equation in \eqref{EP2},
itself is a meromorphic solution of the Poisson equation in
\eqref{EP2}.  Since the solution $M(t)$ in terms of elliptic or hyperbolic functions
is well-known, the solution $P(t)$ can be found by using \eqref{P_mat}.

In the sequel we will regard the generalized Poisson equations in
\eqref{EP2} as a separate system of linear equations
\begin{equation} \label{Poisson}
  \dot\gamma = {\cal A} (t)\gamma, \qquad {\cal A} = k \begin{pmatrix} 0 & a_3 M_3 & -a_2 M_2 \\
    -a_3 M_3 & 0 & a_1 M_1 \\
    a_2 M_2 & -a_1 M_1 & 0 \end{pmatrix}\in \mathrm{so}(3).
\end{equation}
with the coefficients given by the elliptic functions $M(t)$. 
(We will not consider the special cases when $M(t)$ are hyperbolic functions
describing asymptotic motions of the Euler top.)

In the present paper we show that under
the condition $b=ka$ ($k$ is odd), all the solutions of
\eqref{Poisson} are meromorphic. Our main goal is to give an explicit form of
their three independent complex solutions in terms
of elliptic functions and elliptic functions of the second kind
(sigma-functions and exponents), as presented in Theorems \ref{algebraic}
and \ref{comp_sol_theorem} below.  

Additionally, in Theorem \ref{ort_matrix}, we give expressions for the components of the
associated real orthogonal rotation matrix ${\cal R}(t)$ whose columns
satisfy the Poisson equations. 

These equations give rise to a 3rd order ODE for one of the components of the vector $\gamma$. 
In the final part, we compare this ODE with the best known integrable ODE with elliptic coefficients, namely
the Halphen equation, and show that, in general, they cannot be transformed into eauch other.

\section{General properties of the solutions}

As was already mentioned, the Kovalevskaya--Painlev\'e analysis made in \cite{BZ} shows that
for all the solutions of the system \eqref{EP3} or the Poisson equation~\eqref{Poisson}
to be single-valued, the condition~\eqref{cond_k} must hold and $k$ must be an odd integer.
We show that these conditions are also sufficient\footnote{In fact, this
was already stated in \cite{BZ}, however, without a proof.}.
\begin{lemma}
  \label{lem:sv}
  For an arbitrary solution $M_1(t), M_2(t), M_3(t)$ of the
  Euler equations, all solutions of generalized Poisson
  equations~\eqref{Poisson} are single-valued if and only if $k$ is an
  odd integer and condition~\eqref{cond_k} is fulfilled.
\end{lemma}

\begin{proof}
  An elliptic or a hyperbolic solution $M(t)$ of the Euler equation has four simple
  poles in the fundamental region. Hence, all singular points of
  equation~\eqref{Poisson} on $\C$ are regular. Since the equation is linear, branching of its
  solutions can happen only at the singular points.

If $k$ is an odd integer and condition~\eqref{cond_k} is
  satisfied, then all exponents at each singular point are integers.
  However a branching can still occur if the local series solution in a neighborhood of a singular point
has logarithmic terms. We will show that this never happens due to the presence of the first integral \eqref{I4}.
Namely, the integral implies that the equation~\eqref{Poisson} has time dependent first integral
  $I_4(t,\gamma):=\pairing{P(t)}{\gamma}$, which is polynomial of degree
  $k$ in $\gamma$, and $P(t)$ is the corresponding elliptic solution of~\eqref{Poisson}. Assume that $P(t)$ is
normalized: $\langle P(t),P(t)\rangle =1$.

  Now take $t_0\in\C$ which does
  not coincide with a pole of $M(t)$, and a loop  $s\mapsto
  \tau(s)\in\C$, $s\in [0,1]$, $\tau(0)=\tau(1)=t_0$, which encircles once
  counterclockwise a pole $t^*$ of $M(t)$.
  Let $\Gamma(t)$ be a fundamental matrix of~\eqref {Poisson} with the
  first column proportional to $P(t)$ and let $\Gamma(t_0)\in \mathrm{SO}(3,\C)$.
Then $\Gamma(t)\in \mathrm{SO}(3,\C)$ for all $t$ where it is defined.

A continuation along the loop $\tau$ gives a monodromy matrix $ {\cal M}_{\tau} \in \mathrm{SO}(3,\C)$:
$$
 \Gamma(\tau(s+1))= \Gamma(\tau(s))\, {\cal M}_{\tau}.
$$
For any solution $\gamma(t)= \Gamma(t) \vec v$, $\vec v =\text{const}\in {\mathbb C}^3$,
the integral $I_4(t,\gamma)$ implies
$$
\langle P(t_0), \gamma(t_0)\rangle = \langle P(t_0),\Gamma(t_0)\vec v\rangle =
\langle P(t_0),\Gamma(t_0)  {\cal M}_{\tau} \vec v\rangle .
$$
Since $\vec v$ is arbitrary, this yields $ P^T(t_0)\Gamma(t_0) =P^T(t_0)\Gamma(t_0) {\cal M}_{\tau}$ and, due to the
ortogonality of $\Gamma(t)$ and the normalization of $P(t)$,
$$
(1,0,0)=(1,0,0){\cal M}_{\tau} \, .
$$
Then, since, ${\cal M}_{\tau}$ is also orthogonal, it must have the block structure
\begin{equation} \label{mond}
{\cal M}_{\tau} = \begin{pmatrix} 1 & 0 & 0 \\ 0 & \theta & \vartheta \\ 0 & -\vartheta & \theta \end{pmatrix}, \qquad 
\theta, \vartheta\in\C, \quad \theta^2+\vartheta^2=1.
\end{equation}
We now recall that the formal series solution in a neighborhood of the singular point $t^*$
has logarithmic terms if and only if the monodromy matrix ${\cal M}_{\tau}$ is not diagonalizable
(see, e.g., \cite{IY}). However \eqref{mond} is diagonalizable for any $\theta, \vartheta$ satisfying the above condition.
\end{proof}

One of the main tools of our subsequent analysis will be a vector extension of the
known Picard theorem formulated, in particular, in \cite{Fed, Gez}. For our purposes we adopt it in the following form.

\begin{theorem}
  \label{gen_vector_monodromy} Let $T_1$ and $T_2$ be the common, real
  and imaginary periods of the elliptic solutions $M_1(t), M_2(t),
  M_3(t)$ of the Euler equations.  If all the solutions of
  \eqref{Poisson} are meromorphic, then, apart from the elliptic
  vector solution $\gamma(t) =P(M(t))$ of \eqref{Poisson}, there exist
  two elliptic solutions of the second kind $\gamma(t)= G^{(1)}(t)$,
  and $\gamma(t)= G^{(2)}(t)$, which satisfy
  \begin{equation} \label{mon_norm} G^{(1)}(t+T_j)=S_j G^{(1)}(t),
    \quad G^{(2)}(t+T_j)=S_j^{-1} G^{(2)}(t), \qquad j=1,2,
  \end{equation}
  where $S_1, S_2\in {\mathbb C}$, and, moreover, $|S_1|=1$.
\end{theorem}

\begin{proof}The existence of at least one vector solution of the
  second kind, $G(t)$, follows from the vector extension of the Picard
  theorem mentioned above. Let $s_1, s_2$ be its monodromy factors
  with respect to the periods $T_1, T_2$.

  Let $\widehat G(t)$ be another solution of \eqref{Poisson}, and $
  \Gamma(t) = \left(G(t), \widehat G(t), P(t)\right)$ be a fundamental matrix.
  The monodromy matrices $\cM_1$, and $\cM_2$, corresponding to
  periods $T_1, T_2$, respectively, are given by
  \begin{gather*}
    \Gamma (t+ T_j) = \left(s_j G , \, \chi_j G+ \hat \chi_j \hat G+\rho_j P , \, P
    \right)= \Gamma(t) {\cal M}_j, \intertext{where} {\cal M}_j
    = \begin{pmatrix} s_j & \chi_j & 0 \\ 0 & \hat\chi_j & 0 \\ 0 & \rho_j &
      1 \end{pmatrix}, \quad j=1,2,
  \end{gather*}
  where $\chi_j, \hat \chi_j, \rho_j$ are certain constants. Observe that, regardless to the values of the constants,
both monodromy matrices ${\cal M}_1$, and ${\cal M}_2$ are diagonalizable.

  Next, since, by the assumption, all the solutions of \eqref{Poisson}
  are meromorphic, the monodromy group must be trivial.  Therefore,
  ${\cal M}_1$, and ${\cal M}_2$ commute, and are diagonalizable in
  the same basis.  As a result, there exist two independent solutions
  of the second kind $G^{(1)}(t), G^{(2)}(t)$ forming the
  fundamental matrix $\left(G^{(1)}(t), G^{(2)}(t), P(t)\right)$.
Following the general Floquiet theory,
the corresponding monodromy matrices $\bar{\cal M}_1, \bar{\cal M}_2$ must satisfy
$$
\det \bar{\cal M}_j = \exp \left( \int_0^{T_j} \text{Tr} \, {\cal A} (t)
  dt \right)=1, \qquad j=1,2.
$$
(Here we used the property $ {\cal A} (t)\in so(3,{\mathbb C})$.)
Hence, since the monodromy of the elliptic solution $P(t)$ is trivial,
the monodromy factors of $G^{(1)}(t)$, and $G^{(2)}(t)$ are
reciprocal, and this implies \eqref{mon_norm}.

Further, let for certain constants $\nu_1$, $\nu_2\in\C$
$$
\gamma (t) = \nu_1 G^{(1)}(t) + \nu_2 G^{(2)}(t), \quad t\in {\mathbb
  R}
$$
be a real vector solution of the Poisson equation.  This means that, for $i=1,2,3$, the numbers
$\nu_1 G^{(1)}_i(t)$ and $\nu_2 G^{(2)}_i(t)$ are complex conjugated.
Then, for the real period $T_1$, the vector
$\gamma(t+T_1)$ is also a real solution.  On the other hand, from the
above and from the monodromy \eqref{mon_norm}, we deduce that
$$
\gamma_i(t+T_1) = S_1 \, \nu_1 G^{(1)}_i(t+ T_1) + S_1^{-1} \, \nu_2
G^{(2)}_i (t+T_1), 
$$
which is real if and only if $|S_1|=1$.
\end{proof}
\section{Algebraic parametrization and elliptic sigma-function
  solution for $M$ and $P(M)$.}

We first recall how generic solutions of the Euler equation in \eqref{EP3} can
be expressed in terms of the Weierstrass sigma functions.  We need
this fact to derive the general solution of the Poisson equation.

Let us fix a common level of first integrals \eqref{EI}
\begin{equation}
  \label{EIlev}
  \langle M ,a M \rangle =l,\qquad
  \langle M , M \rangle=m^2, \qquad  \langle \gamma, \gamma
  \rangle =1
\end{equation}
For a generic values $\chi$, $m$, solutions $M_i(t)$ of the Euler
equations are elliptic functions related to the elliptic curve $E$,
given by
\begin{equation}
  \label{E0} E=\left\{\mu^2 = U_4 (\lambda) \right\}, \qquad
  U_4(\lambda):= -(\lambda -a_{1})(\lambda -a_{2})(\lambda-a_{3})(\lambda -c) ,
\end{equation}
where $c:= l/m^2$. Here and below we assume that $c\ne a_1, a_2, a_3$. This curve, compactified and regularized, 
has two infinite points $\infty_{\pm}$.

A "rational" parametrization of the momenta $M_{i}$ in terms of the
coordinates $\lambda$, see, e.g., \cite{Acta_bill}, have the following
form
\begin{gather}
  M_\alpha = m\sqrt{(a_\beta-c)(a_\gamma-c) \over
    (a_\alpha-a_\beta)(a_\alpha-a_\gamma)} \sqrt{ \frac{\lambda
      -a_\alpha}{\lambda -c} }.  \label{par0}
\end{gather}
Then, from the Euler equations, we easily deduce that the evolution of
$\lambda$ is given by the equation
\begin{equation}
  \label{dot_la}
  \dot \lambda= 2 m \sqrt{ -(\lambda
    -a_{1})(\lambda -a_{2})(\lambda -a_{3})(\lambda -c) }.
\end{equation}
That is, for any $\lambda\in {\mathbb C}$, the right hand sides of
\eqref{par0} satisfy the equations \eqref{EIlev}.

For a real motion, i.e., for real values of $l$, $m$, and $t$, if
$a_1< a_2 <a_3$, then one has $c\in (a_1,a_3)$, $c\ne a_2$. Moreover,
\begin{equation*}
  \lambda \in
  \begin{cases}
    [a_2,a_3],  & \mtext{if}  a_1<c<a_2, \\
    [a_1,a_2], & \mtext{if} a_2<c<a_3 .
  \end{cases}
\end{equation*}

The birational map $(\lambda,\mu)\to (z,w)$, given by
\begin{gather}
  z= \frac 13 \frac{(\tau_2 -2c \tau_1 +3c^2)\la+ 2c \tau_2
    -c^2\tau_1-3 \tau_3}{\la-c}, \quad w= \frac{\mu}{ (\la-c)^2
  }, \label{BT}
\end{gather}
transforms the elliptic curve $E$ to its canonical Weierstrass form
\begin{gather} \label{canon_W}
  {\cal E} =\left\{w^{2}=U_3(z) \right\}, \qquad U_3(z):= 4 ( z-e_{1})
  (z-e_{2} )(z-e_{3})=4z-g_2z-g_3
\end{gather}
where
\begin{equation}
  \begin{split}
    3 e_\alpha = \tau_2+c\tau_1-3(a_\beta a_\gamma+ c a_\alpha) , \qquad  e_1+e_2+e_3=0, \\
    \tau_1 = a_1+a_2+a_3, \qquad \tau_2= a_1a_2+a_2 a_3+a_3a_1, \qquad
    \tau_3= a_1 a_2 a_3.
  \end{split}
\end{equation}
The above map sends $\lambda=c$ to $z=\infty$, and $a_i$ to $e_i$,
respectively. Then there is the following relation between the
holomorphic differentials on $E$ and $\cal E$:
$$
\rmi \frac{ \rmd \la }{2\sqrt{ U_4( \la)} } =\frac{ \rmd z }{\sqrt{4
    (z-e_1)(z-e_2)(z-e_3) } }.
$$
Let us introduce the Abel map
\begin{equation} \label{A1} u = \rmi \int_c^{p} \frac{ \rmd \la }{2
    \sqrt{ U_4( \la)} },\, \mtext{where} p=(\la,\mu) \in E.
\end{equation}
The integrals
$$
\Omega_\alpha:=\rmi \int_c^{a_\alpha} \frac{ \rmd \la }{2 \sqrt{ U_4(
    \la)} }, \qquad \alpha=1,2,3
$$
are the half-periods of the curve $E$. We choose the sign of the root
$U_4( \la)$ to ensure $\Omega_1+ \Omega_2+\Omega_3=0$.

According to \eqref{A1}, in the case $a_1 < c < a_2 < a_3$ the
half-period $\Omega_1$ is imaginary and $\Omega_2$ is real, whereas
for $a_1 < a_2 < c < a_3$ the half-period $\Omega_3$ is imaginary and
$\Omega_2$ is real. In both cases, comparing \eqref{A1} with
\eqref{dot_la}, we get
\begin{equation} \label{time} u= \rmi m (t-t_0)+ \Omega_2.
\end{equation}

Using the Weierstrass sigma function $\sigma(u)=\sigma(u |\,
2\Omega_1, 2\Omega_3)$, one can write
\begin{gather}
  \la-c = \text{const } \cdot \frac{\sigma^2 (u)}{\sigma(u-h)\,
    \sigma(u+h)}, \quad \text{where} \quad h=\int_c^{\infty}\frac{d\la }{\sqrt{ U_4( \la)} }, \\
  \frac{\la-\rho}{\la-c} = \text{const }\cdot \frac{
    \sigma(u-\beta)\sigma(u+\beta)}{\sigma^2(u)}, \quad \beta
  =\int_c^{\rho} \frac{ \rmd \la }{\sqrt{ U_4( \la)} } . \label{wf_2}
\end{gather}
Moreover, we also have
\begin{equation}
  \label{s/s}
  \sqrt{\frac{\lambda -a_\alpha}{\lambda -c}} =
  C_\alpha \frac{\sigma_\alpha(u)}{\sigma(u)},  \quad
  \alpha=1,2,3,
\end{equation}
where $C_{\alpha}$ are certain constants and $\sigma_\alpha(u)$ are the
sigma-functions obtained from $\sigma(u)$ by shift of $u$, and by
multiplication by an exponent:
\begin{equation} \label{sig} \sigma_\alpha(u):= \rme^{\eta_\alpha u}
  \frac{\sigma(\Omega_\alpha-u) }{\sigma(\Omega_\alpha)}, \qquad
  \eta_\alpha=\zeta(\Omega_\alpha), \quad
  \zeta(u)=\frac{\sigma'(u)}{\sigma(u)},
\end{equation}
where $\alpha=1,2,3$.  Note that we have
\begin{equation} \label{exp_sig} \sigma(u)=u -\frac{g_2}{240} u^5 -
  \frac{g_3}{840} u^7+ \cdots, \quad \text{and} \quad \sigma_1(0) =
  \sigma_2(0) =\sigma_3(0)=1.
\end{equation}
see, e.g., \cite{H_Cu} or \cite{Law}.  From \eqref{par0}, \eqref{s/s},
it follows that the solutions of the Euler equations have the form
\begin{equation}
  \label{sols_M}
  M_\alpha = h_\alpha
  \frac{\sigma_\alpha (u)}{\sigma(u)}, \quad u= \rmi  m\, t+ \Omega_2 , \quad \alpha=1,2,3.
\end{equation}
with certain constants $h_\alpha$ which we determine below. In view of
\eqref{sig}, the sigma-quotients have the quasiperiodic property
\begin{equation} \label{quasi}
\frac{\sigma_\alpha(u+ 2\Omega_j)}{
    \sigma(u+ 2\Omega_j)} = (-1)^{1-\delta_{\alpha,j}}
  \frac{\sigma_\alpha(u)}{\sigma(u)},
\end{equation}
where $\delta_{\alpha,j}$ is the Kronecker symbol. Hence, the
coefficients $M_\alpha(u)$ of the Poisson equation \eqref{Poisson}
have common periods $4\Omega_1, 4\Omega_2$.

Next, using the parametrization \eqref{par0}, and the expressions
\eqref{P_mat}, for each odd $k$ we get the following parametrization
for the elliptic solution $P(M)=(P_1, P_2, P_3)^T$:
\begin{equation}
  \label{ell_P}
  P_\alpha= 
  \sqrt{(a_{\beta}-c)(a_{\gamma}-c) \over (a_\alpha-a_\beta)(a_\alpha-a_\gamma)} \sqrt{
    \frac{\lambda -a_\alpha}{\lambda -c} } \frac{ F_{s,\alpha}(\la)}{(\la-c)^s},
\end{equation}
where $F_{s,\alpha}(\la)= \rho_{0,\alpha}\prod_{r=1}^{s} \left( \lambda -  \rho_{r,\alpha}\right)$
is a polynomial of degree $s=(k-1)/2$, which is obtained by substituting \eqref{par0} into the vector
$\Phi_k T$ in \eqref{P_mat} and taking the numerator. 
The sum $\Delta_k = P_1^2(\la)+P_2^2(\la)+ P_2^2(\la)$ is a constant depending on $a_\alpha$, and $c$ only.

In particular, for $k=3$, and $b_\alpha = 3 a_\alpha$, by using \eqref{eq:2}, we have
\begin{align}
  F_{11}(\lambda) & =[3\tau_2 +4c(c-\tau_1-2a_1) -2 a_2a_3]
  \lambda + c(\tau_{2}+2 a_{2}  a_{3}) -4\tau_3 + 8 c^{2} a_{1} , \notag \\
  F_{12}(\lambda) & =[3\tau_2 +4c(c-\tau_1-2a_2) -2 a_3a_1]
  \lambda + c(\tau_{2}+2 a_{3}  a_{1}) -4\tau_3 + 8 c^{2} a_{2} ,  \label{FF} \\
  F_{13}(\lambda) & = [3\tau_2 +4c(c-\tau_1-2a_3) -2 a_1a_2] \lambda +
  c(\tau_{2}+2 a_{1} a_{2}) -4\tau_3 + 8 c^{2} a_{3}\notag
\end{align}
and $\Delta_3 = \tau_2^{2} - 4\tau_{1}\tau_3+ 36c \tau_3- 48c^{2}
\tau_2 + 64c^{3}\tau_1$.

Now, applying expressions \eqref{wf_2}, \eqref{s/s} to \eqref{ell_P},
we get
\begin{equation}
  \label{sols_P}
  P_\alpha = c_\alpha
  \frac{\sigma_\alpha(u)}{\sigma(u)}
  \prod_{r=1}^s \frac{\sigma(u+
    v_{r,\alpha}) \sigma
    (u-v_{r,\alpha})}{\sigma^{2}(u)},
\end{equation}
where
\begin{equation}
  v_{r,\alpha} = \pm \rmi \int_c^{\rho_{r,\alpha}}  \frac{ \rmd \la
  }{2\sqrt{ U_4( \la)} }, \qquad u=\rmi  m t+ \Omega_2, \notag
\end{equation}
for $r=1, \dots, s$. Thus, the components $P_\alpha(u)$ have a pole of
order $k$ at $u=0$ and, like $M_\alpha(u)$, they are doubly periodic
with common periods $4\Omega_1, 4\Omega_2$. We finally have

\begin{proposition} \label{M_P} The momentum vector $M$ and the elliptic vector
  solution $P$ of the Poisson equations can be written as
  \begin{gather}
    \label{sols_Pn}
    M_\alpha = m \epsilon_\alpha \frac{\sigma_\alpha (u)}{\sigma(u)},
    \quad P_\alpha = \epsilon_\alpha
    \frac{\sigma_\alpha(u)}{\sigma(u)}
    \prod_{r=1}^s \frac{\sigma(u+v_{r,\alpha})\, \sigma(u- v_{r,\alpha} )}{\sigma^2( v_{r,\alpha})\, \sigma^{2}(u)}, \\
    \epsilon_{\alpha } =\frac{1}{\sqrt{(a_\alpha -a_\beta)(a_\alpha-a_\gamma)}}, \quad
    (\alpha,\beta,\gamma)=(1,2,3), \label{eps's}
  \end{gather}
  where signs of $\epsilon_{\alpha}$ are chosen according to the condition
$$
 \frac{1}{\epsilon_1 \epsilon_2 \epsilon_3} = - (a_1-a_2)(a_2-a_3)(a_3-a_1),
$$ 
and $u$ depends on time $t$ via \eqref{time}.

  Then also
  \begin{equation}\label{P^2}
    P_\alpha^2(u) = \epsilon_\alpha \left(\wp(u)-\wp(\Omega_\alpha) \right) \prod_{l=1}^s \left(\wp(u)-\wp(v_{l,\alpha }) \right)^2,
  \end{equation}
  where $\wp(u)=\wp(u|g_2, g_3)$ is the Weierstrass $P$-function.

  Here, for any $u\in {\mathbb C}$
  \begin{gather}
    M_1^2(u)+ M_2^2(u)+M_3^2(u)= m^2, \quad P_1^2(u)+
    P_2^2(u)+P_3^2(u)=\Pi, \label{PP} \\ \notag
\Pi = \epsilon_\beta^2
    (\wp(\Omega_\alpha)-\wp(\Omega_\beta)) \prod_{l=1}^s
    \left(\wp(\Omega_\alpha)-\wp(v_{l,\beta }) \right)^2 \\
    \qquad + \epsilon_\gamma^2 (\wp(\Omega_\alpha)-\wp(\Omega_\gamma))
    \prod_{l=1}^s \left(\wp(\Omega_\alpha)-\wp(v_{l,\gamma })
    \right)^2 , \label{Pi}
  \end{gather}
  for any permutation $(\alpha,\beta,\gamma)=(1,2,3)$.
\end{proposition}

\paragraph{Remark.} According to the rule \eqref{quasi}, the shift $u \to u+ 2\Omega_\alpha$ in the solutions
 \eqref{sols_Pn} is equivalent to flip of signs of some of the constants $\epsilon_i$ in such a way that 
the above condition is satisfied.   

\begin{proof}[Proof of Proposition \ref{M_P}] To calculate the constants $h_\alpha, c_\alpha$ in the
  elliptic solutions \eqref{sols_M}, \eqref{sols_P}, we compare the
  leading terms of their Laurent expansions near the poles and the
  expansions of the sigma functions. Namely, let $t_0\in {\mathbb C}$
  be a pole of the functions $M(t), P(t)$, and $\delta t=
  t-t_0$. Substituting
$$
M=\frac {1}{\delta t} \left( M^{(0)}+M^{(1)}\delta t+\cdots\right) ,
\quad P= \frac {1}{(\delta t)^k} \left( P^{(0)}+P^{(1)} \delta
  t+\cdots\right)
$$
into the equations \eqref{EP2} for $M$ and $\gamma$, for any $k\in
{\mathbb N}$, one gets
\begin{gather} \label{lead_term_M} M^{(0)} \in\left\{ \rmi
    (\epsilon_{1},\epsilon _{2},\epsilon _{3})^{T}, \; \rmi (-\epsilon
    _{1},-\epsilon _{2},\epsilon _{3})^{T}, \; \rmi (-\epsilon
    _{1},\epsilon _{2},-\epsilon _{3})^{T}, \; \rmi (\epsilon
    _{1},-\epsilon _{2},-\epsilon _{3})^{T} \right\} ,
\end{gather}
with $\epsilon _{\alpha}$ given by \eqref{eps's}, and $P^{(0)}$ is
proportional to $M^{(0)}$.

On the other hand, in view of \eqref{exp_sig}, near $u=0$ we have the expansions
\begin{gather*}
  \frac { \sigma_\alpha(u)}{\sigma(u)}= \frac{1}{u}+O(1), \\
  \prod_{r=1}^s \frac{\sigma(u+ v_{r,\alpha}) \sigma(u-v_{r,\alpha})}{\sigma^{2}(u)} = -
  \frac{\sigma^2(v_{1,\alpha})\cdots \sigma^2(v_{s,\alpha}) }{u^{k-1}}+O(1).
\end{gather*}
with $s=(k-1)/2$. Since in the above expansions $u=\rmi m \cdot \delta
t$, comparing them, we obtain\footnote{ If fact, one can write
  $h_\alpha, c_\alpha$ only in terms of sigma-constants and
  $\sigma(v_j)$, as it was written for the Euler top (the case $k=1$)
  (see \cite{Jac, Whitt}), but this process is tedious and requires
  more calculations.}
$$
h_\alpha = m \epsilon_\alpha, \quad
  c_\alpha = \frac{\epsilon_{\alpha}}{ \sigma^2(v_1) \cdots
    \sigma^2(v_s)} .
$$
Substituting this into \eqref{sols_M}, \eqref{sols_P}, we get
\eqref{sols_Pn}.

The latter, in view of the known relations (see, e.g., \cite{H_Cu,Law})
$$
\frac{\sigma_\alpha^2(u)} {\sigma^2(u)}= \wp(u)-\wp(\Omega_\alpha) ,
\quad \frac{\sigma(u+\beta) \sigma(u-\beta) }{ \sigma^2(\beta)\,
  \sigma^{2}(u)}= \wp(u)-\wp(\beta) ,
$$
implies \eqref{P^2}.

Finally, since $P(u)$ is a solution of the Poisson equations, it
satisfies the integral \eqref{PP}. Setting there $u=\Omega_\alpha$ and
using \eqref{P^2} one obtains \eqref{Pi}.
\end{proof}

\paragraph{Remark.} As follows from the formal Laurent solution for
$M(t)$ with the coefficients \eqref{lead_term_M}, near a pole $t=t_0$
{\it any} vector solution $\gamma(t)=(\gamma_1, \gamma_2, \gamma_3)^T$ of the Poisson equation
\eqref{Poisson} has the expansion
\begin{equation} \label{exp_G} \gamma_\alpha (t) = \frac
  {\text{const}}{(\delta t)^k} \left(\epsilon_\alpha + O(\delta t)
  \right) , \quad \alpha=1,2,3.
\end{equation}

\section{Algebraic structure of elliptic solutions of 2nd kind}
Using the algebraic parameterizations \eqref{par0} and \eqref{ell_P},
we obtain

\begin{theorem} \label{algebraic}
  1) If $k$ is a positive odd integer and $k\ge 3$, then, apart from
  the solution $P(\la)$ in \eqref{ell_P}, the Poisson equations
  \eqref{Poisson} has two independent solutions
$$
\gamma^{(j)} =G^{(j)}= \left(G_1^{(j)},G_2^{(j)},G_3^{(j)}\right)^T,
\qquad j=1,2
$$
which can be represented as the following algebraic functions of the
parameter $\la$ in \eqref{par0}, \eqref{dot_la}
\begin{equation}
  \label{g12}
  \begin{split}
    G_{\alpha}^{(1)} &= c_{1,\alpha} \frac {\sqrt{ Q_{k,\alpha} (\la)}
    }{\sqrt{(\la-c)^k} } \exp
    \left( \, \frac 12 \int W_\alpha \right), \\
    G_{\alpha}^{(2)}&= c_{2,\alpha} \frac {\sqrt{ Q_{k,\alpha} (\la)}
    }{\sqrt{(\la-c)^k} } \exp \left( - \frac 12 \int W_\alpha \right)
    .i
  \end{split}
\end{equation}
Here $c_{1,\alpha}, c_{2,\alpha}$ are certain constants to be
specified below, and
\begin{align}
  \label{Ints}
  W_\alpha & = \frac{ q_{s+1,\alpha} (\la) \cdot (\la-c)^{s} }
  { Q_{k,\alpha} (\la)}\, \frac{ \rmd \la }{ \sqrt{ U_4( \la)}  }, \\
  Q_{k,\alpha}(\la) & = r_{0,\alpha}\prod_{i=1}^{k}(\la-r_{i,\alpha})
  = \textup{const} \cdot (P_\beta^2(\la)+ P_\gamma^2(\la))\cdot (\la-c)^k \qquad \qquad \notag \\
  & = (a_\beta-c)(a_\gamma-a_\beta)(\la-a_\beta) F_{s,\beta}^2(\la) +
  (a_\gamma-c)(a_\gamma-a_\beta)(\la-a_\gamma)
  F_{s,\gamma}^2(\la), \label{QQ} \\
  & \qquad \qquad (\alpha,\beta,\gamma)=(1,2,3), \notag
\end{align}
where the polynomials $F_{s,\alpha} (\la)$ of degree $s=(k-1)/2$ are
specified in~\eqref{ell_P}.

The above formula implies that the zeros of the polynomials
$Q_{k,\alpha}(\la)$ coincide with the zeros of
$P_\beta^2(\la)+P_\gamma^2(\la)$.  \medskip

2) The differential $W_\alpha$ is a meromorphic differential of the
third kind on $\cal E$ having pairs of only simple poles at the points
$ {\cal P}_{i,\alpha}^{\pm} =(r_{i,\alpha}, \pm \sqrt{
  U_4(r_{i,\alpha})})\in {\cal E}$, $i=1,\dots,k$ with residues $\pm
1$ respectively:
\begin{equation} \label{res1} \Res_{ {\cal P}_{i,\alpha}^{\pm} }
  W_{\alpha} = \pm 1, \qquad i= 1,\dots,k .
\end{equation}
Finally, $q_{s+1,\alpha} (\la)$ in \eqref{Ints} are polynomials of
degree $s+1=(k+1)/2$ completely defined by the conditions \eqref{res1}
\end{theorem}

The algebraic solutions in the classical case $k=1$ will be described
separately in Section \ref{k=1}.

\paragraph{Remark.}
The polynomials $F_{s,1},F_{s,2}, F_{s,3}$ and $Q_{k,1},Q_{k,2},
Q_{k,3}$ are obtained by the corresponding permutation of
$a_1,a_2,a_3$.  Note that their Abel images ($u$-coordinates) of their
roots $\rho_{r,\alpha}, r_{i,\alpha}$ are not obtained from each other
by the translations by the half-periods $\Omega_j$ of the elliptic
curve $E$.  \medskip

\begin{proof}[Proof of Theorem~\ref{algebraic}]
  1) According to the kinematic interpretation, the Poisson equations
  in \eqref{EP2} describes the evolution of a fixed in the space
  vector $\gamma$ in a frame rotating with the angular velocity
  $\tilde\omega= B M$ (also taken in the body frame). Now choose a
  fixed in space ortonormal frame $\{ O, \be_1, \be_2, \be_3\}$.  Let
  $\theta,\psi,\phi$ be the Euler nutation, precession, and rotation
  angles associated to this frame so that the corresponding rotation
  matrix is
$$
{\cal R}= \begin{pmatrix}
  \cos\phi\cos\psi - \cos\theta \sin\psi \sin\phi & \cos\phi\sin\psi + \cos\theta \cos\psi \sin\phi & \sin\phi \sin\theta \\
  -\sin\phi\cos\psi - \cos\theta \sin\psi \cos\phi & -\sin\phi\sin\psi + \cos\theta\cos\psi\cos\phi & \cos\phi \sin\theta \\
  \sin\theta \sin\psi & -\sin\theta \cos\psi & \cos\theta
\end{pmatrix}.
$$
Let $\vec P(t)=(P_1, P_2, P_3)^T$ be a solution of \eqref{EP2}
describing the motion of the vector $|P| {\bf e}_3$. Then, in view of
the structure of $\cal R$, and from the Euler kinematic equations, one
has
$$
P_1 = |P|\sin\phi \sin\theta, \quad P_2 = |P|\cos\phi \sin\theta,
\quad P_3=|P| \cos\theta,
$$
and
\begin{equation} \label{dot_psi} \dot \psi = |P| \frac { \tilde
    \omega_1 P_1 + \tilde \omega_2 P_2}{P_1^2+ P_2^2},
\end{equation}
see e.g., \cite{Whitt}. Hence the thirds components of the other two
independent solutions $ G^{(1)}$, and $G^{(2)}$ of \eqref{Poisson} can
be written in the complex form
\begin{equation} \label{det_angular} G_3^{(j)} = \tilde c_{j}
  \sin\theta \exp(\pm \rmi \psi), \quad \text{or} \quad
  G_3^{(j)}=c_{j} \sqrt{P_1^2+ P_2^2} \exp\left(\pm \rmi \int\dot
    \psi\, \rmd t \right) ,
\end{equation}
$\tilde c_{j}$, and $c_{j}$ are certain constants, and $j=1,2$.  Using
the parametrization \eqref{ell_P} for $P_\alpha(\la)$, as well as
relations \eqref{par0} and \eqref{dot_la}, we get
\begin{equation} \label{angle}
\rmi\dot \psi \, dt = \rmi \sqrt{\Delta_k} k \frac{ a_1 M_1(\la)
  P_1(\la) + a_2 M_2(\la) P_2(\la) }{P_1^2(\la)+ P_2^2(\la) } \,
\frac{d\la}{2 m \sqrt{U_4(\la)}} := \frac 12 W_3 .
\end{equation}
After simplifications this takes the form
$$
\rmi\dot \psi \, dt = \frac 12 \frac{ q_{s+1,3} (\la) \cdot
  (\la-c)^{s} } { Q_{k,3} (\la)}\, \frac{ d \la }{ \sqrt{ U_4( \la)}
},
$$
with
\begin{gather*}
  q_{s+1,3} (\la) = \frac{\rmi\sqrt{\Delta_k}\, k}{m}
  [ (a_2-c)(a_3-a_2)b_1 \cdot (\la-a_1) F_{s,1}(\la) \\
  \hskip 6cm + (a_1-c)(a_1-a_3)b_2 \cdot (\la-a_2) F_{s,2}(\la) ] , \\
  Q_{k,3} (\la) = (a_2-c)(a_3-a_2)\cdot (\la-a_1) F_{s,1}^2(\la)+ (a_1-c)(a_1-a_3)\cdot(\la-a_2) F_{s,2}^2(\la), \\
  \sqrt{P_1^2+ P_2^2} = \text{const}\, \frac {\sqrt{ Q_{k,3} (\la)}
  }{\sqrt{(\la-c)^k} } .
\end{gather*}
The above implies the formulas \eqref{g12}--\eqref{QQ} for $\alpha=3$.
Repeating the same geometric argumentation for $\alpha=1,2$, we get
the whole set of formulas of Theorem \ref{algebraic}.  \medskip

2) The differential $W_\alpha$ in \eqref{Ints} has simple poles at
$\la\in \{r_{1,\alpha},\dots, r_{k,\alpha}\}$, and each of them
corresponds to two points ${\cal P}_{i,\alpha}^{\pm}$ on $E$. In view
of the degrees of polynomials $Q_{k,\alpha}(\la)$, and
$q_{s+1,\alpha}$, this differential does not have poles at the
infinite points $\infty_\pm$ on $E$. Next, $\rmd \la/\sqrt{U_4( \la)}$
does not have neither poles nor zeros on $E$.  Hence $W_\alpha$ is a
differential of the third kind.

Next, the residuum conditions \eqref{res1} are necessary for the
solutions \eqref{g12} to be meromorphic in $t$, or $u$ and, locally,
in $\la$. Namely, let $\tau=\la-r_{i,\alpha}$ be a local coordinate on
$E$ near the root $r_{i,\alpha}$ and the meromorphic differentials
have the expansion
$$
W_\alpha = \left( \frac{\varkappa}{\tau}+ O(1) \right) \rmd\tau.
$$
Assume $\varkappa>0$. Then, as follows from \eqref{det_angular} for $\alpha=3$, the
leading term of the expansion of the solution $\Gamma_3$ has the
form
$$
\text{const}\cdot \sqrt{\tau} \exp\left( \frac{\varkappa}{2} \ln \tau
\right ) = \text{const}\cdot \sqrt{\tau} \, \tau^{\varkappa/2}.
$$
Hence, $\varkappa$ must be 1 or $3,5,\dots$. Since $\Gamma_3^{(1,2)}$ is an
elliptic function of the second kind, the total number of its zeros on
$E$ must be equal to that of its poles (with multiplicity), that is,
$k$, therefore the residuum $\varkappa$ must be 1. The same argumentation for $\alpha=1,2$ completes the proof.
\end{proof}

\section{Sigma-function solutions of 2nd kind}
In order to convert the algebraic solutions of Theorem \ref{algebraic}
to analytic ones, we shall need the following formula.
\begin{proposition} \label{3rd_kind} Let $K_{k}(\la)$ be a
  polynomial of odd degree $k$, and
$$
W = \frac{K_{k}(\la)}{Q_k(\la)}\,
\frac{ \rmd \la }{2 \sqrt{ U_4( \la)} }, \qquad Q_k(\la)= r_0(\la-r_1)\cdots (\la-r_k)
$$
be a differential of the third kind on the degree 4 curve $E$ with simple poles
at the points ${\cal P}_{j}^{\pm} =(r_{j}, \pm \sqrt{U_4(r_{j})})$,
$j=1,\dots,k$ with residues $\pm 1$ respectively.  Let the point $(\la,\mu)\in E$ and
$u\in {\mathbb C}$ be related by the Abel map \eqref{A1}. Then
\begin{gather}
  \int_{(c,0)}^{(\la,\mu)} W = \log \frac { \sigma (u- w_1)\cdots
    \sigma(u-w_k)} {\sigma (u+w_1)\cdots \sigma(u+w_k) }
  + 2 [ \zeta(w_1) +\cdots+ \zeta(w_k) ]u + \delta \, u -\pi \,\rmi, \label{del'} \\
  \frac {\sqrt{ (\la-r_1)\cdots (\la-r_k) } }{\sqrt{(\la-c)^k}} =
  \textup{const} \frac{ \sqrt{ \sigma (u- w_1)\cdots \sigma(u-w_k)
      \,\sigma (u+ w_1)\cdots \sigma(u +w_k) } }{\sigma^k (u) },
  \label{s_roots}
\end{gather}
where, as above, $\zeta(u)$ is the Weierstrass zeta-function, $\delta=K(c)/Q(c)$, and
\begin{equation} \label{w_j}
w_j = \rmi \int_{(c,0)}^{{\cal P}_{j}^{-}} \frac{\rmd \la }{2\sqrt{U_4( \la)} }\,.
\end{equation}
The correct signs of the roots $w_j$ can be chosen from the conditions
\begin{equation} \label{sign_w_j}
 \frac{K_{k}(r_j)}{(r_j-c) Q'_k(r_j)}= - \frac{3 \rmi}{(a_1-c)(a_2-c)(a_3-c)} \wp'(w_j).
\end{equation}
\end{proposition}

The proposition is a reformulation of known relations in the theory of
elliptic functions, its proof is purely technical and given in Appendix 2.

Note that if the polynomial $K_k(\la)$ contains the factor $(\la-c)^s$, $s\ge 1$, the constant
$\delta$ in \eqref{del} is zero.

Theorem \ref{algebraic} and Proposition \ref{3rd_kind} allow us to
formulate the following theorem.
\begin{theorem} \label{comp_sol_theorem} 1) The two complex vector
  elliptic solutions of the second kind of the Poisson equations are
  \begin{gather} \label{2nd_kind_sol_}
    G^{(1)}(u)=\left(G_1^{(1)} ,G_2^{(1)} , G_3^{(1)} \right)^T, \qquad
G^{(2)}(u)=\left(G_1^{(2)} ,G_2^{(2)} , G_3^{(2)}\right)^T, \\
    G_\alpha^{(1)} (u) = \epsilon_\alpha \rme^{\Theta_\alpha u}
    \prod_{l=1}^{k} \frac{ \sigma(u-w_{l,\alpha})}
    {\sigma(u)\,\sigma(-w_{l,\alpha})} , \qquad G_\alpha^{(2)} (u) =
    \epsilon_\alpha \rme^{-\Theta_\alpha u} \prod_{l=1}^{k} \frac{
      \sigma(u+w_{l,\alpha})} {\sigma(u)\,\sigma(w_{l,\alpha})} , \notag \\
w_{l,\alpha}= \rmi \int_{c}^{ r_{l,\alpha} } \frac{\rmd \la }{2\sqrt{U_4( \la)} }\, , \qquad l=1,\dots,k,
\quad \alpha =1,2,3, \notag
  \end{gather}
where $r_{l,\alpha}$ are the roots of the polynomials $Q_\alpha(\la)$ in \eqref{QQ}, 
the signs of $w_{j,i}$ are defined according to \eqref{w_j}, \eqref{sign_w_j}. Next,  
$\epsilon_{\alpha }$ are specified in \eqref{eps's}, $u=\rmi m t+ \Omega_2$, and
  \begin{gather*}
    \Theta_1 = \sum_{j=1}^k \zeta(w_{j,1}), \quad \Theta_2 =
    \sum_{j=1}^k \zeta(w_{j,2}), \quad \Theta_3 =\sum_{j=1}^k
    \zeta(w_{j,3}).
  \end{gather*}
  Together with \eqref{sols_Pn}, \eqref{eps's}, the expressions
  \eqref{2nd_kind_sol_} form a complete basis of independent solutions
  of the equations \eqref{Poisson}. 
\medskip

  2) Let also
$$
\Sigma_1 = \sum_{j=1}^k w_{j,1}, \quad \Sigma_2 = \sum_{j=1}^k
w_{j,2}, \quad \Sigma_3 =\sum_{j=1}^k w_{j,3}.
$$
The solutions \eqref{2nd_kind_sol_} have the
quasi-monodromy 
\begin{gather} \label{mon_alpha} G_\alpha^{(1)} (u+ 2\Omega_j) =
  (-1)^{\delta_{\alpha j}} s_j \, G_\alpha^{(1)} (u), \quad
  G_\alpha^{(2)} (u+ 2\Omega_j) = (-1)^{\delta_{\alpha j}} s_j^{-1}\, G_\alpha^{(2)} (u),    \\
  j=1,2,3, \quad \alpha = 1,2,3 \notag
\end{gather}
and imply the vector monodromy
\begin{equation} \label{vec_momodromy} G^{(1)}(u+ 4\Omega_j)= s_j^2
  G^{(1)} (u), \quad G^{(2)}(u+ 4\Omega_j)= s_j^{-2} G^{(2)}(u),
\end{equation}
where
\begin{align}
  s_1 & = - \exp( 2 \Theta_1\Omega_1-2 \Sigma_1\eta_1)= \exp( 2
  \Theta_2\Omega_1-2 \Sigma_2 \eta_1)=
  \exp( 2 \Theta_3\Omega_1-2 \Sigma_3 \eta_1) , \notag \\
  s_2 & =- \exp( 2 \Theta_2\Omega_2-2 \Sigma_2\eta_2)= \exp( 2
  \Theta_1\Omega_2-2 \Sigma_1 \eta_2)
  = \exp( 2 \Theta_3\Omega_2-2 \Sigma_3 \eta_2), \label{s_j} \\
  s_3 & = - \exp( 2 \Theta_3\Omega_3- 2\Sigma_3\eta_3)= \exp( 2
  \Theta_2\Omega_3-2 \Sigma_2 \eta_3) = \exp( 2 \Theta_1\Omega_3-2
  \Sigma_1 \eta_3), \notag
\end{align}
3) If $\Omega_j$ is the imaginary half-period, then $|s_j|=1$. For the
real half-period $\Omega_2$ one has $|s_2| \ne 1$.  Moreover,
\begin{equation} \label{diff_Sigma_Theta}
\Sigma_\alpha-\Sigma_\beta
  =\Omega_\gamma \quad \textup{mod} \; \{ 2\Omega_1{\mathbb
    Z}+2\Omega_2{\mathbb Z}\}, \quad \Theta_\alpha-\Theta_\beta
  =\eta_\gamma \quad \textup{mod} \; \{ 2\eta_1{\mathbb
    Z}+2\eta_2{\mathbb Z}\}.
\end{equation}
for $(\alpha,\beta,\gamma)=(1,2,3)$.  \medskip

\noindent 4) Finally, for any $u\in {\mathbb C}$,
\begin{gather} \label{zero_sum}
  \pairing{ G^{(1)}(u)} {G^{(1)}(u)} =0, \qquad \pairing{ G^{(2)}(u)}{G^{(2)}(u)} =0, \\
 G_\alpha^{(1)}(u) G_\alpha^{(2)}(u)= - \epsilon_\alpha^2 \prod_{r=1}^k (\wp(u)-\wp(w_{r,\alpha})), \label{G12}
\end{gather}
and
\begin{gather} \label{G^2}
  \begin{aligned}
    \left[ G_\alpha^{(1)} (\Omega_j)\right]^2 & = s_j \;
    (-1)^{1-\delta_{\alpha j}} \epsilon_i^2 \prod_{r=1}^k
    (\wp(\Omega_j)-\wp(w_{r,\alpha})),  \\
    \left[ G_\alpha^{(2)} (\Omega_j)\right]^2 & = s_j^{-1}
    (-1)^{1-\delta_{\alpha j}} \epsilon_i^2 \prod_{r=1}^k
    (\wp(\Omega_j)- \wp(w_{r,\alpha})),
  \end{aligned} \\
  \alpha,j=1,2,3.  \notag
\end{gather}
Moreover, for any $u\in {\mathbb C}$,
\begin{align}
 - \langle G^{(1)} (u), G^{(2)}(u)\rangle & =
\epsilon_1^2 \prod_{l=1}^k (\wp(\Omega_1)-\wp(w_{l,1})) \notag \\
  = \epsilon_2^2 \prod_{l=1}^k (\wp(\Omega_2)-\wp(w_{l,2}))
 & = \epsilon_3^2 \prod_{l=1}^k (\wp(\Omega_3)-\wp(w_{l,3}))=
  - \Pi, \label{eq_Pi}
\end{align}
where the constant $\Pi$ is defined in \eqref{PP}, \eqref{Pi}.
\end{theorem}

The proof of the Theorem can be found in Appendix 2. 

\paragraph{Remark.} One can easily recognize that, for each index $\alpha$,
the components $G^{(1)}_\alpha (u), G^{(2)}_\alpha (u)$ have the same
structure as solutions of the Lame equation
$$
  \frac{d^2\Lambda}{d\, u^2}=( n(n+1)\wp(u)+ B )\Lambda, \quad n\in{\mathbb N}, \quad B=\text{const}
$$
with $n=k$ (see, e.g., \cite{WhW}), namely,
$$
 \Lambda_1= \prod_{s=1}^k \left(\frac{\sigma (u-h_s)}{\sigma(u) \, \sigma(h_s)} \exp(\zeta(h_s)u )\right), \quad
\Lambda_2= \prod_{s=1}^k \left(\frac{\sigma (u+h_s)}{\sigma(u) \, \sigma(h_s)} \exp(-\zeta(h_s)u )\right),
$$
where the zeros $h_1,\dots,h_k$ satisfy various conditions, in particular,
$$
\wp(h_1)+ \cdots + \wp(h_k)= k B.
$$
However, as numerical tests show, the zeros $w_{1,\alpha}, \dots, w_{k,\alpha}$ of the solutions \eqref{2nd_kind_sol_} do not satisfy all the conditions on $h_1,\dots,h_k$. Hence $G^{(1)}_\alpha(u), G^{(2)}_\alpha(u)$ cannot be solutions of the Lame equation.

Thus, if the relation between the Poisson equations \eqref{Poisson} and the Lame equation (or some of its generalizations) exists, it should be a rather non-trivial one.


\section{The classical case $k=1$}\label{k=1}
The Poisson equations in this case were first integrated by C. Jacobi \cite{Jac},
who used previous results of Legendre (see, e.g., \cite{Whitt}).
The case does not fit completely into Theorems \ref{algebraic}, \ref{comp_sol_theorem} because the corresponding
meromorphic differentials \eqref{Ints} do not contain the factor $(\la-c)^s$, and the solutions do not have precisely
the structure of \eqref{2nd_kind_sol_}.

Namely, now the elliptic solution $P$ is just $M(u)$ given by \eqref{sols_Pn} and the algebraic solutions \eqref{g12} reread
$$
G_\alpha^{(1,2)} = c_\alpha \sqrt{M_\beta^2+ M_\gamma^2} \exp \left(\pm \frac 12 \int W_\alpha \right), \qquad
(\alpha,\beta,\gamma)=(1,2,3).
$$
Set, for concreteness, $\alpha=3$. Using the algebraic parameterization \eqref{par0} for $M(\la)$,
from \eqref{angle} we get
\begin{align}
W_3 &= 2\rmi \frac{a_1 M_1^2(\la)+ a_2 M_2^2(\la) }{ M_1^2(\la)+ M_2^2(\la) } dt \notag \\
 & = 2 \frac{ (ca_3-a_1 a_2) \la +c(a_1a_2-a_1 a_3-a_2 a_3)-a_1a_2 a_3}{(c+a_3-a_1-a_2)\la-(ca_3-a_1 a_2)}\,
\frac{d\la}{2 \sqrt{U_4(\la)}}.  \label{W33}
\end{align}
This is a differential of 3rd kind having a pair of simple poles $(\la^*,\pm \sqrt{U_4(\la^*)})$ on $E$ with
$$
\la^*=\frac{a_1 a_2-c a_3}{c+a_3-a_1-a_2} .
$$
Observe that in \eqref{W33}
$$
 2\frac{ (ca_3-a_1 a_2) \la +c(a_1a_2-a_1 a_3-a_2 a_3)-a_1a_2 a_3}{(c+a_3-a_1-a_2)\la-(ca_3-a_1 a_2)}\bigg |_{\la=c}=2 a_3.
$$
Then, according to Proposition \ref{3rd_kind},
\begin{gather*}
  \int_{(c,0)}^{(\la,\mu)} W_3 = \log \frac {\sigma (u-w_3)} {\sigma (u+w_3)}
  + 2 [ \zeta(w_3)+ a_3 ]u - \pi \rmi, \\
\sqrt{ M_1^2+ M_2^2} = \text{const}\, \frac{\sqrt{\sigma(u-w_3) \, \sigma(u+w_3) } }{\sigma(u)} ,
\end{gather*}
where $w_3$ is the Abel image of the pole $(\la^*,\mu)$ of $W_3$ with residuum $-1$,
$$
 w_3 = \int_{(c,0)}^{\la^*} \frac{\rmi \, d\la}{2\sqrt{U_4(\la)}} .
$$

As a result, up to multiplication by $-1$, the elliptic solutions of 2nd kind are
\begin{equation} \label{G12_k1}
\begin{aligned}
 G_\alpha^{(1)}(u) &= \epsilon_\alpha \frac{\sigma(u-w_\alpha)}{\sigma(u) \,\sigma(-w_\alpha)}
\exp{[(\zeta(w_\alpha)+a_\alpha)u]}, \\
G_\alpha^{(2)}(u) & = \epsilon_\alpha \frac{\sigma(u-w_\alpha)}{\sigma(u) \,\sigma(-w_\alpha)}
\exp{[-(\zeta(w_\alpha)+a_\alpha)u]}
\end{aligned}
\end{equation}
(compare with \eqref{2nd_kind_sol_}),
where $w_\alpha$ denote the Abel image of the pole of the differential $W_\alpha$ with the residuum $-1$, and,
as above, $u=\rmi m t+\Omega_2$.

As follows from item 3 of Theorem \ref{comp_sol_theorem}, here
\begin{gather*}
\begin{aligned}
w_\alpha-w_\beta & =\Omega_\gamma \quad \textup{mod} \;
\{ 2\Omega_1{\mathbb Z}+2\Omega_2{\mathbb Z}\}, \\ 
\zeta(w_\alpha)-\zeta(w_\beta)
  & =\eta_\gamma=\zeta(\Omega_\gamma) \quad \textup{mod} \; \{ 2\eta_1{\mathbb Z}+2\eta_2{\mathbb Z}\} , 
\end{aligned} \\
(\alpha,\beta,\gamma)=(1,2,3).
\end{gather*}
Then, introducing
$$
 w^*= w_\alpha-\Omega_\alpha, \quad  \theta^*=\zeta( w_\alpha)-\eta_\alpha=\zeta(w^*) \quad \text{for any $\alpha$}
$$
and using the definition \eqref{sig} of the sigma-functions with indices, one can represent the complex solutions
\eqref{G12_k1} in the following form
\begin{equation} \label{G12_k}
\begin{aligned}
 G_\alpha^{(1)}(u) & = \epsilon_\alpha \frac{\sigma_\alpha(u-w^*)}{\sigma(u) \, \sigma_\alpha(-w^*) }
\exp{[(\theta^*+a_\alpha)u]},  \\
 G_\alpha^{(2)}(u) & = \epsilon_\alpha \frac{\sigma_\alpha(u+w^*)}{\sigma(u) \, \sigma_\alpha(w^*) }
\exp{[-(\theta^*+a_\alpha)u]}.
\end{aligned}
\end{equation}
By transforming the integrals defining $w_\alpha$, one can show that
$$
 w^* = \int_{(c,0)}^{\infty} \frac{\rmi \, d\la}{2\sqrt{U_4(\la)}} ,
$$
where $\infty$ stands for one of the two infinite points on $E$.

Finally notice that being rewritten in terms of theta-functions, the expressions \eqref{G12_k}
coincide with the complex solutions presented by C. Jacobi (see also \cite{Whitt}).

\section{Real normalized vector solutions.} As was shown in Section 2,
the elliptic solution $P(u)$ in \eqref{sols_Pn}, \eqref{eps's} for
$u=\rmi nt+\Omega_2$, $t\in {\mathbb R}$, is real.  Then, by their
construction (see \eqref{det_angular}), the basis elliptic 2nd kind
solutions $G_i^{(1)}(u), G_i^{(2)}(u)$ have opposite arguments:
\begin{equation} \label{argum} \text{Arg} \left( G_\alpha^{(1)}
    (u)\right) = -\text{Arg}\left( G_\alpha^{(2)}(u)\right)
\end{equation}
for any $\alpha =1,2,3$. Then two independent {\it non-normalized}
real vector solutions can be written in the form
\begin{equation}
  \begin{aligned}
    \gamma^{(1)}(t) &= \nu_1 G^{(1)}(\rmi  m t+ \Omega_2) + \nu_2 G^{(2)}(\rmi  m t+ \Omega_2),  \\
    \gamma^{(2)}(t) & = \frac {1}{\rmi } \left[\nu_1 G^{(1)}(\rmi m t+
      \Omega_2) - \nu_2 G^{(2)}(\rmi m t+ \Omega_2)\right]\, ,
  \end{aligned} \label{real_sols}
\end{equation}
with some appropriate constants $\nu_1, \nu_2$. In view of
\eqref{argum}, it is sufficient to set
\begin{equation} \label{nu's} \nu_1= \chi | G_\alpha^{(1)}(\rmi m t^*+
  \Omega_2)|^{-1}, \quad \nu_2= \chi |G_\alpha^{(2)}(\rmi m t^*+
  \Omega_2)|^{-1},
\end{equation}
for any fixed real $t^*$, any real nonzero $\chi$, and any $\alpha\in
\{1,2,3 \}$. Then we arrive at

\begin{theorem} \label{ort_matrix} A real orthogonal rotation matrix
  formed by the three independent unit vector solutions of the Poisson
  equations \eqref{Poisson} has the form
  \begin{gather} {\cal R}(t) = \frac{1}{\sqrt{\Pi}} \left(
      \bar\gamma^{(1)}(t) , \bar\gamma^{(2)}(t) , \bar P(t) \right),
    \qquad
    \bar P(t) = P(\rmi  m t+ \Omega_2), \label{ort_R} \\
    \bar\gamma^{(1)} (t) = \frac 12 \left[ \frac{1}{\sqrt{s_2}}
      G^{(1)} (\rmi m t+ \Omega_2)
      + \sqrt{s_2}\, G^{(2)}(\rmi  m t+ \Omega_2) \right ], \notag \\
    \bar\gamma^{(2)} (t) = \frac 12 \left[ \frac{1}{\sqrt{-s_2}}
      G^{(1)}(\rmi m t+ \Omega_2) + \sqrt{-s_2}\, G^{(2)}(\rmi m t+
      \Omega_2) \right ] , \notag
  \end{gather}
  where $P(u)$ is the elliptic solution \eqref{sols_Pn},
  \eqref{eps's}, and $G^{(1)}(u), G^{(2)}(u)$ are the elliptic
  solutions of the second kind described in \eqref{2nd_kind_sol_} and
  Theorem \ref{comp_sol_theorem}. Next, $s_2$ is real and is specified
  in \eqref{s_j}, whereas the constant $\Pi$ is defined in \eqref{PP},
  \eqref{Pi}.
\end{theorem}

Note that the columns of ${\cal R}(t)$ form a left- or right-oriented
orthonormal basis.  \medskip

\noindent{\it Proof of Theorem} \ref{ort_matrix}. Setting in
\eqref{nu's} $t^*=0, \alpha=2$, and
$$
\chi= \frac{\epsilon_2}{2}\sqrt{(\wp(\varOmega_2)-
  \wp(\alpha_{1,2}))\cdots (\wp(\varOmega_2)- \wp(\alpha_{k,2}))}\, ,
$$
in view of \eqref{G^2}, we get $\nu_1= s_2^{-1/2},
\nu_2=s_2^{1/2}$. Then \eqref{real_sols} gives
\begin{align*}
\gamma^{(1)}(0) & = \frac 12 \left[\frac{1}{ \sqrt{s_2} }
  G^{(1)}(\Omega_2)+ \sqrt{s_2}G^{(2)}(\Omega_2)\right], \\
\gamma^{(2)}(0) &= \frac 12 \left[\frac{1}{ \sqrt{-s_2} }
  G^{(1)}(\Omega_2)+ \sqrt{-s_2}G^{(2)}(\Omega_2)\right].
\end{align*}
Then, in view of \eqref{zero_sum} and \eqref{eq_Pi}
\begin{align}
  \sum_{\alpha=1}^3 \left[ \gamma^{(1)}_\alpha (0)\right]^2 & =
  \sum_{\alpha=1}^3 \left[\gamma^{(2)}_\alpha (0)\right]^2
  = - \frac 12 \langle G^{(1)} (\Omega_2), G^{(2)}(\Omega_2)\rangle = \Pi .
\end{align}

As a result, the real vectors $\bar\gamma^{(1)}(t), \bar\gamma^{(2)}(t),
\bar P(t)$ all have the same length. By their construction, they are
all orthogonal. Hence, we obtain the matrix ${\cal R}(t)$ in \eqref{ort_R}.  $\square$

\section{The case of negative odd $k$.}
We first note that the case of negative odd $k$ cannot be reduced to the already considered case $k>0$ by
the trivial substitution $t=-T$ in the Poisson equations in \eqref{EP2}. Indeed, this change gives
$$
\frac{d\gamma}{dT}= -k \gamma \times a M(-T)
$$
The elliptic vector solution $M(-T)$ given by \eqref{sols_Pn} with $u=-\rmi m T+\Omega_2$ is neither odd nor even, hence
one cannot write $M(-T)=M(T)$, and the above equation cannot be transformed to the form
$$
\frac{d\gamma}{dT}= -k \gamma \times a M(T).
$$

Nevertheless, the analysis for positive $k$ is sufficient to cover
all the cases. Indeed, upon introducing new moments of inertia
\begin{equation} \label{nn} \Lambda_i= (J_j+J_k-J_i) J_i, \quad
  (i,j,k)=(1,2,3) ,
\end{equation}
the system \eqref{EP2} can be rewritten as
\begin{equation} \label{EP4} \dot \Lambda\omega = - \Lambda\omega
  \times \omega, \quad \dot\gamma = k \gamma \times \omega ,
\end{equation}
which, under the change $t= -T$, gives
\begin{equation} \label{EP5} 
\Lambda\omega' = \Lambda\omega \times
  \omega, \quad \gamma' = - k \gamma \times \omega , \qquad (')=\frac{d}{dT}.
\end{equation}
The Euler equations here have the integrals
$$
\langle \omega ,\Lambda \omega \rangle=L, \quad  \langle \omega,\Lambda^2 \omega \rangle ={\cal M}^2 
$$
with integration constants $L,{\cal M}$. 
Then, applying to these equations the procedure of section 3, we express the solutions $\omega(T)$ in terms elliptic functions of the curve
$$E'=\{W^2=-(Z-A_1)(Z-A_2)(Z-A_3)(Z-C)\}$$ with the parameters
\begin{gather*}
  A_i=1/\Lambda_i= \frac{a_i^2 a_j a_k}{a_i a_j+a_i a_k-a_j a_k}, \\
  C= \frac{L}{{\cal M}^2}= \frac{(\tau_2\, c- 2\,\tau_3)\,\tau_3}
{\tau_2^2-4\tau_1\tau_3 + 4\,\tau_3\,c}\, , \quad
{\cal M}^2 = \langle \omega,\Lambda^2 \omega \rangle = \frac{\tau_2^2-4\tau_1\tau_3 + 4\,\tau_3\,c}{\tau_3^2} m^2,  \\
\tau_1 = a_1+a_2+a_3, \qquad \tau_2= a_1a_2+a_2 a_3+a_3a_1, \qquad
    \tau_3= a_1 a_2 a_3 .
\end{gather*}
Here, as above, in \eqref{eq:1}, 
$$
a_i=1/J_i, \quad c= l/m^2, \quad m^2 = \langle M,M\rangle=\langle \omega, J^2 \omega\rangle,
$$

Since $\omega(T)$, like $M(t)$, must also be elliptic functions of the original curve $E$, we get
the following relation 
\begin{align} 
dt & = \frac{d\la}{2m \sqrt{-(\la-a_1)(\la-a_2)(\la-a_3)(\la-c)}} \notag \\
& = - \frac{d Z}{2{\cal M} \sqrt{- (\la-A_1)(\la-A_2)(\la-A_3)(\la-C) }}.     \label{la_Z}
\end{align}

\paragraph{Remark.}
As one may expect, the elliptic curves $E, E'$ with the parameters $a_i,c$ and
$A_i,C$ are birationally equivalent.  Indeed, $E'$ is transformed to $E$ by the substitution
\begin{align*}
Z & = \frac{(\tau_2\, \lambda - 2\,\tau_3)\,\tau_3}
{\tau_2^2-4 \tau_1\tau_3 + 4\,\tau_3\,\lambda},  \\
W & =\mu \frac{ (a_1 a_2- a_2 a_3-a_1 a_3)(a_1 a_3- a_2 a_3-a_2 a_3)(a_2 a_3- a_1 a_3-a_1 a_2)}
{ (\tau_2^2-4 \tau_1\tau_3 + 4\,\tau_3\,\lambda)^2 \, (\tau_2^2-4 \tau_1\tau_3 + 4\,\tau_3\,c) }.
\end{align*}
We stress that the half-periods of $E$ and $E'$
\begin{gather*}
\Omega_\alpha= \int_c^{a_\alpha} \frac{\rmi  \rmd \la }{2 \sqrt{ U_4(\la)} }, \quad
\Omega_\alpha'=\int_C^{A_\alpha} \frac{\rmi \rmd Z}{2 \sqrt{-(\la-A_1)(\la-A_2)(\la-A_3)(\la-C)} }, \\
\alpha=1,2,3
\end{gather*}
in general, do not coincide, but only proportional to each other: in view of \eqref{la_Z}, \\
$\Omega_\alpha'=\frac{\cal M}{m} \Omega_\alpha$, $\alpha=1,2,3$.

Now comparing \eqref{EP2} or \eqref{EP4} with \eqref{EP5}, we arrive at the following observation. 

\begin{proposition} Let $k$ be an odd negative integer and the vector \\ $\gamma(T)=\gamma(T|A_1,A_2,A_3,C,{\cal M})$ 
be a solution of the Poisson equations in \eqref{EP5} with the elliptic coefficients $\omega(T)$ related to the
parameters $A_i,C,{\cal M}$. Then $\gamma(t)=\gamma(-T)$ is a solution of the Poisson equations
  \eqref{Poisson} with the elliptic oefficients $M_\alpha(t)$ related to the parameters $a_i, c, m$, and vice versa.
\end{proposition}

In other words, the solutions $\gamma(t)$ of \eqref{Poisson} with an odd negative $k$ 
and the parameters $a_1,a_2,a_3,c,m$ are given by $\gamma(-t|A_1,A_2,A_3,C,{\cal M})$.  
The latter are described by the formulae of Theorems \ref{algebraic} and \ref{comp_sol_theorem} corresponding to $|k|$,
the parameters $A_1,A_2,A_3,C,{\cal M}$, and the corresponding roots of the polynomials
$F_{s,\alpha}(\la), Q_{|k|,\alpha}(\la)$.


We stress that although the elliptic vector functions $\omega(t)$ for $a_1,a_2,a_3,c,m$ and $\omega(-t)$ for 
$A_1,A_2,A_3,C,{\cal M}$, coincide, this is no more true for the solutions $\gamma(t|a_1,a_2,a_3,c,m)$ and  
$\gamma(-t|A_1,A_2,A_3,C,{\cal M})$. 

\section{Comparison with the Halphen equation}
By using the algebraic parametrization \eqref{par0}, the generalized Poisson equations \eqref{Poisson} can be rewritten
as 3rd order ODE for one of the components of the vector $\gamma$, say $\gamma_1$, with the independent variable 
$\lambda\in {\mathbb C}$.
For general $a_i,k$ the explicit expressions for the coefficients of the ODE are very long.
We give an example for $a_1=1, a_2=2, a_3=3$ and $k=3$ ($c$ is arbitrary): 
\begin{gather}
\frac{d^3}{d\la^3} \gamma_1+ g_2(\la) \frac{d^2}{d\la^2} \gamma_1 +
 g_1(\la)\frac{d}{d\la}\gamma_1 + g_0(\la) \gamma_1=0, \label{3rd_order}
\end{gather}
with
\begin{align*}
g_2 & = \frac {3}{2\,( \la - 2)} + \frac {3}{2\,(\la -3)} + \frac {3}{\la -c}
 - \frac {c + 10}{10\,\la - 15\,c  + 6 + c\,\la } + \frac {1}{\la  - 1}, \\
g_1 & = - \frac {3\,(138\,c - 37)}{52\,(c - 2)\,( \la - 2)} + \frac {82\,c - 119}{24\,( c- 3)\,( \la - 3)}
 - {\displaystyle \frac {34\,c^{2} - 276\,c + 263}{4\,(c - 2)\,(c - 3)\,(c - 1)\,(c - \la )}}  \\
& \mbox{} - \frac {1}{4\,(\la  - 1)^{2} }  +
{\displaystyle \frac {(10 + c)^{2}\,(10\,c^{2} + 2267\,c - 2666)
}{156\,(c - 2)\,( - 3 + c)\,(7\,c - 8)\,( - 15\,c + 10\,\la  + 6 + c\,\la )} }  \\
& - {\displaystyle \frac {5}{4\,(c - \la )^2 } }  +
{\displaystyle \frac {254\,c^{2} - 669\,c + 436}{8\,(7\,c - 8)\,(c- 1)\,( \la - 1)} }, \\
g_0 & = - {\displaystyle \frac {27\,(c - 1)\,(23\,c - 30)}{104\,( \la - 2)\,(c - 2)^{2}}}
 + {\displaystyle \frac {3\,(c - 1)\,(17\,c - 22)}{32\,(\la  - 3)\,(c - 3)^{2}}}  + {\displaystyle
\frac {9\,(4\,c^{3} - 81\,c^{2} + 171\,c - 98)}{8\,(c - \la )\,( c- 1)^{2}\,(c - 2)^{2}\,(c - 3)^{2}}}  \\
& \mbox{} + {\displaystyle \frac {9\,(5\,c - 6)}{16\,( c- 1)\,(
 - 1 + \la )^{2}}}  - {\displaystyle \frac {3\,( - 1 + c)\,(10 + c)
^{5}}{208\,(c - 2)^{2}\,( - 3 + c)^{2}\,(7\,c - 8)\,( - 15\,c + 10\,\la  + 6 + c\,\la )}}  \\
& + \frac {9}{8\,(c - 3)\,( c- 1)\,(c-\la )^{2}} +
\frac {9\,(109\,c^{3} - 235\,c^{2} + 106\,c + 24)}{32\,( - 1 + c)^{2}\,(7\,c - 8)\,( \la - 1)} .
\end{align*}
That is, the coefficients have poles at $\la =a_1, a_2, a_3, c$,
and at an extra pole defined by the condition $- 15\,c +6+ (10+c)\,\la =0$.

A natural question is how the above 3rd order equation is related with known linear equations with elliptic
coefficients admitting elliptic solutions of second kind. The best known example is the {\it Halphen equation}
\begin{equation} \label{Halp}
\frac{d^3}{d u^3} \Psi + (1-n^2)\wp(u) \frac{d}{d u}\Psi+ \wp'(u) \frac{1-n^2}{2} \Psi =h \Psi, \qquad \Psi=\Psi(u),
\end{equation}
where $n$ is integer and $h$ is an arbitrary parameter. As above, $\wp(u)$ is the Weierstrass function.
For any such $n$
the 3 independent solutions $\Psi_1(u), \Psi(u), \Psi_3(u)$ are elliptic functions of 2nd kind with poles of order
$g=n-1$ at $u=0$:
$$
\Psi_\alpha( u ) = \frac {\sigma (u- w_1^{(\alpha)}(h))\cdots \sigma(u-w_g^{(\alpha)}(h))} {\sigma^g (u) } \,
\exp [\; (\zeta(w_1^{(\alpha)}) +\cdots+ \zeta(w_g^{(\alpha)})) u ], \qquad \alpha=1,2,3.
$$
The structure of the solutions generalizes that of solutions \eqref{2nd_kind_sol_} of our equation
\eqref{3rd_order}.
So, one can suppose that the equation \eqref{3rd_order} is a special case of the Halphen equation (for $h=0$),
when one of its solution is elliptic.

However, written in the algebraic form with the independent variable $z$ such that
\begin{equation*}
u = \int_z^\infty \frac {dz}{2 \sqrt{ (Z-e_{1}) (Z-e_{2}) (Z-e_{3})} }, \quad \wp(u)=z
\end{equation*}
the Halphen equation with $h=0$ is
\begin{gather*}
4(z-e_1)(z-e_2)(z-e_3) \frac{d^3}{d z^3} \Psi + g_2(z) \frac{d^2}{d z^2} \Psi +
 g_1(z) \frac{d}{d z}\gamma_1 + g_0(z) \Psi =0,  \\
 g_2= 36 z^2+6 e_1 e_2+6 e_1 e_3+6 e_2 e_3, \quad g_1=12 z-1-n^2, \quad g_0= -(1-n^2)/2.
\end{gather*}
Hence, the coefficients of the normalized equation have finite poles only at $z=e_1, e_2, e_3$.
Taking into account that the equation \eqref{3rd_order} has 5 poles
(which can be reduced to 4 finite poles), it cannot be identified with the special case of the Halphen equation.

\section*{Acknowledgments}
The work of Yu.F was supported by the MICIIN grants MTM2009-06973 and MTM2009-06234. 
A.J.M. and M.P. acknowledge the support of grant 
DEC-2011/02/A/ST1/00208 of National Science Centre of Poland.

\section*{Appendix 1: A numerical example.} This example was made with
{\tt Maple} by using the functions {\tt WeierstrassP(u,g[2],g[3]),
  WeierstrassSigma(u,g[2],g[3]), WeierstrassZeta(u,g[2],g[3])}.
Consider the simplest nontrivial case $k=3$ and $B=ka$.  Choose the
inertia tensor $J=\diag (1,1/2, 1/3)$, i.e., $a=\diag (1,2,3)$ and
$l/n^2=c=5/2\in [a_2, a_3]$, $n=1$. Therefore, in the real case the
parameter $\la\in [a_1, a_2]=[1;2]$. The elliptic curve $E$ has the
form
$$
\mu^2 = 2 U_4 (\la) = -2 (\la-1)(\la-2)(\la-3)(\la-5/2).
$$
The birational transformation \eqref{BT}, namely,
$$
z= -\frac {1}{12} \frac{\la +2} {z-5/2}, \quad \la = 2 \frac{15
  z-1}{12 z+1}
$$
takes it to the Weierstrass form
\begin{gather*}
  w^2 =  4z^3-\frac{7}{3} z+\frac{10}{27} = 4(z+ 5/6)(z-1/6)(z-2/3), \\
  \text{with} \qquad z(\la=1)=1/6, \quad z(2)=2/3, \quad z(3)=-5/6.
\end{gather*}
so that the parameters of the Weierstrass functions of $E$ are $g_2=
7/3$, $g_3= -10/27$.  The half-periods are\footnote{In this example
  most of the float numbers are indicated up to $10^{-6}$.}:
\begin{gather*}
  \Omega_1 = \rmi  \int_c^{a_1} \frac{d\la}{\mu} = \int_\infty^{1/6} \frac{d z}{w} = -1.656638- 1.415737 \cdot \rmi ,  \\
  \Omega_2 = \int_\infty^{2/3} \frac{d z}{w}= 1.656638, \quad \Omega_3
  = \int_\infty^{-5/6} \frac{d z}{w}= 1.415737 \cdot  \rmi .
\end{gather*}
The corresponding constants $\eta_i=\zeta(\Omega_i)$ in \eqref{sig}
are
$$
\eta_1 = -0.4402056+0.57199 \cdot \rmi , \quad \eta_2=0.4402056, \quad
\eta_3=- 0.57199 \cdot \rmi .
$$
This allows to calculate 
$$
\sigma_\alpha (u) = \exp (\eta_\alpha u )\, \frac{ \sigma(\Omega_\alpha- u|g_2, g_3) }
{\sigma(\Omega_\alpha |g_2, g_3) }, \qquad \alpha=1,2,3.
$$
Next,
$$
\epsilon_1 = \frac{1}{\sqrt{(a_1 -a_2)(a_1-a_3)}} =
\frac{\sqrt{2}}{2}, \quad \epsilon_2 = -\rmi , \quad \epsilon_3 =
\frac{\sqrt{2}}{2}.
$$
From \eqref{FF} we have
$$
F_{11}(\la) = -34 \la +167/2, \quad F_{12} = -48\la +237/2, \quad
F_{13}= -66 \la+ 327/2 ,
$$
The Abel images of their zeros on the complex plane $u$ are
$$
v_1= \pm 0.34497195, \quad v_2 = \pm 0.2898, \quad v_3 = \pm 0.24686
\, .
$$
Then the {\it real} elliptic solutions for the Euler and the Poisson
equations are given by
\begin{gather}
  M_\alpha = \epsilon_\alpha \frac{\sigma_\alpha (u)}{\sigma(u)},
  \quad P_\alpha = \epsilon_\alpha
  \frac{\sigma_\alpha (u) \sigma (u-v_\alpha) \sigma (u+v_\alpha)}{\sigma^3(u)}, \label{example_P} \\
  \alpha =1,2,3 , \quad u= \rmi t + \Omega_2 , \qquad t\in {\mathbb
    R}, \notag
\end{gather}
with the indicated above values of the parameters. Note that here
$|M|=1$ and $|P|^2=\Pi=201.062507$.

Next, the meromorphic differentials in \eqref{Ints} are
\begin{align*}
  W_1 &= \rmi \frac{3 \sqrt{3217}
    \,(2y-5)(260\,y^2-1295\,y+1613)}
  {2(3984\,y^3-29608\,y^2+73363\,y-60607) \sqrt{U_4(\la)}}, \\
  W_2 &= \rmi
  \frac{3\sqrt{3217}\,(2\,y-5)(560\,y^2-3136\,y+4331)}
  {2(23824\,y^3-194232\,y^2+523569\,y-467236) \sqrt{U_4(\la)}}, \\
  W_3 &= \rmi
  \frac{3\sqrt{3217}\,(2\,y-5)(1084\,y^2-4913\,y+5521)}
  {2(50672\,y^3-356280 \,y^2+ 832461\,y-646139) \sqrt{U_4(\la)}}.
\end{align*}
Then, up to constant factors, the polynomials $Q_{3,\alpha}$ in
Theorem \ref{algebraic} are
\begin{align*}
  Q_{3,3} (\la) & = 3167/2\,\la^3-44535/4\,\la^2+832461/32\,\la-646139/32, \\
  Q_{3,2} (\la)  & =1489\,\la^3-24279/2\,\la^2+523569/16\,\la-116809/4, \\
  Q_{3,1}(\la) & =
  6723/2\,\la^3-99927/4\,\la^2+1980801/32\,\la-1636389/32 \, .
\end{align*}
The Abel images of their zeros on the complex $u$-plane are
\begin{gather*}
  w_{1,3} =\pm 0.21963966 , \quad w_{2,3} =\pm 0.309571742 , \quad
  w_{3,3}= \pm 1.232049297, \\
  w_{1,2} =\pm (\Omega_3+0.40885), \quad
  w_{2,2}, w_{3,2}=\pm (0.272475 \pm 0.041941\cdot \rmi ), \\
  w_{1,1} =- 0.26888528, \quad w_{2,1}, w_{3,1}=\pm (0.3424574955\pm
  0.2549408658\cdot \rmi ).
\end{gather*}
Applying the condition \eqref{sign_w_j}, we choose the signs
\begin{align}
  w_{1,3} &= - 0.21963966 , \quad w_{2,3} =- 0.309571742 , \quad w_{3,3}= 1.232049297, \notag \\
  w_{1,2} &=-0.40885 + \Omega_3, \quad w_{2,2}= -0.272475 + \Omega_3, \notag \\ 
w_{3,2}& = -0.272475-\Omega_3, \label{ww's}  \\
  w_{1,1} & =- 0.26888528, \quad w_{2,1}=-0.3424574955- 0.2549408658\cdot \rmi , \notag \\
  w_{3,1} & = - 0.3424574955 + 0.2549408658\cdot \rmi . \notag
\end{align}
Then
$$
\Sigma_3 = \sum_{j=1}^3 w_{j,3}= 0 .7028378946, \quad \Sigma_2 =
-0.95380028+\Omega_3, \quad \Sigma_1 = -0.95380028
$$
and
$$
\Theta_3= \sum_{j=1}^3 \zeta(w_{j,3}) =-7.03775, \quad
\Theta_2=-7.477955 -0.572\cdot \rmi , \quad \Theta_1= -7.477955 \, .
$$
Notice that
\begin{equation} \label{diff_sigma} \Sigma_3-\Sigma_2 = 1.656638-
  1.41573\cdot \rmi = \tilde\Omega_1 \equiv \Omega_1, \quad
  \Sigma_2-\Sigma_1=\Omega_3, \quad \Sigma_3-\Sigma_1 = \Omega_2
\end{equation}
and
\begin{equation} \label{diff_theta} \Theta_3-\Theta_2 =
  \zeta(\tilde\Omega_1), \quad \Theta_2-\Theta_1 = \zeta(\Omega_3) =
  \eta_3 , \quad \Theta_3-\Theta_1 = \zeta(\Omega_2)=\eta_2.
\end{equation}
That is, the sums of zeros of $Q_1(\la), Q_2(\la), Q_3(\la)$ on the
complex $u$-plane differ by the half-periods of the curve $E$, as predicted by item 3 of Theorem \ref{comp_sol_theorem}.
Note that for arbitrary values of $w_{i,j}$, the relations
\eqref{diff_sigma} does not imply \eqref{diff_theta}.

The basis complex vector elliptic solutions of the 2nd kind in
\eqref{2nd_kind_sol_} are
\begin{gather} 
  G_\alpha^{(1)} (u) = \, \epsilon_\alpha \frac{ \sigma(u-w_{1, \alpha
    } )\, \sigma(u-w_{2,\alpha})\, \sigma(u-w_{3,\alpha}) }
  {\sigma^3(u)\,\sigma(-w_{1,\alpha})\, \sigma(-w_{2,\alpha})\, \sigma(-w_{3,\alpha})} e^{\Theta_\alpha u}, \notag \\
  G_\alpha ^{(2)} (u) = \, \epsilon_\alpha \frac{ \sigma(u+w_{1,\alpha
    })\, \sigma(u+w_{2,\alpha })\, \sigma(u+w_{3,\alpha }) }
  {\sigma^3(u)\,\sigma(w_{1,\alpha })\, \sigma(w_{2,\alpha })\, \sigma(w_{3,\alpha })} e^{-\Theta_\alpha  u}, \notag \\
  \alpha =1,2,3, \notag
\end{gather}
with $w_{l,\alpha}$ specified in \eqref{ww's}. Next, 
\begin{align*}
  \begin{pmatrix} G_1^{(1)} \\G_2^{(1)} \\ G_3^{(1)}\end{pmatrix}
  (u+2\Omega_3)
  & = s_3 \begin{pmatrix} G_1^{(1)} \\ G_2^{(1)} \\  - G_3^{(1)}\end{pmatrix} (u),  \\
  s_3 & = \exp (2[\Theta_1\Omega_3-\Sigma_1\eta_3]  ) = -0.962799+0.2702178\cdot \rmi , \quad |s_3|=1, \\
  \begin{pmatrix} G_1^{(1)} \\G_2^{(1)} \\ G_3^{(1)}\end{pmatrix}
  (u+2\Omega_2) & = s_2 \begin{pmatrix} G_1^{(1)} \\ - G_2^{(1)} \\
    G_3^{(1)}\end{pmatrix} (u), \quad s_2 = \exp
  (2[\Theta_1\Omega_2-\Sigma_1\eta_2] ) = 0.402144 \cdot 10^{-10}.
\end{align*}
The monodromy of the solutions $G_i^{(2)}(u)$ is inverse to the above one.

Finally, the ortogonal matrix of real solutions is
\begin{gather*}
{\cal R}(t) = \frac{1}{\sqrt{\Pi}} \left(
      \bar\gamma^{(1)}(t) \; \bar\gamma^{(2)}(t) \; \bar P(t) \right)
\qquad
    \bar P(t) = P(\rmi t+ \Omega_2), \\
    \bar\gamma^{(1)} (t) = \frac 12 \left[ \frac{1}{ \sqrt{s_2} }
      G^{(1)} (\rmi t+ \Omega_2)
      + \sqrt{s_2}\, G^{(2)}(\rmi t+ \Omega_2) \right ], \notag \\
    \bar\gamma^{(2)} (t) = \frac 12 \left[ \frac{1}{\sqrt{-s_2}}
      G^{(1)}(\rmi t+ \Omega_2) + \sqrt{-s_2}\, G^{(2)}(\rmi t+ \Omega_2) \right ] , \notag
  \end{gather*}
with $\sqrt{s_2}=0.634148 \cdot 10^{-5}$ and $\Pi=201.0625$.

\section*{Appendix 2: Proofs of Proposition \ref{3rd_kind} and Theorem \ref{comp_sol_theorem}.}
Proposition \ref{3rd_kind} is a reformulation of
known relations of the theory of elliptic functions.  Consider first
an elliptic curve $\cal E$ in the canonical Weierstrass form \eqref{canon_W},
$$
{\cal E} =\left\{w^{2}=P_3(z)\equiv 4 ( z-e_{1}) (z-e_{2} )(z-e_{3})
\right\}, \quad e_1+e_2+e_3=0,
$$
a point $P=(z,w)=(z,\sqrt{P_3(z)})$ on it, and the Abel map
\begin{equation} \label{A} u = \int_\infty^{P} \frac{ dz }{2\sqrt{(
      z-e_{1}) (z-e_{2} )(z-e_{3}) } }\, dz ,
\end{equation}
which gives $z=\wp(u)$. 

Now let
\begin{equation}\label{Omega}
  \bar W = \frac{ \bar q_k(z) }{2\sqrt{P_3(z)} \, \bar Q_k(z) }\, dz , \qquad \bar Q_k (z) =(z-z_1)\cdots (z-z_k),
\end{equation}
be a meromorphic differential of 3rd kind having
pairs of only simple poles at the finite points $ {\cal P}_i^{\pm}
=(z_i, \pm 2\sqrt{ R_3(z_i)})$, $i=1,\dots,k$ with residia $\pm 1$
respectively. Here $q_k(z) = b_k z^k+ \cdots + b_0$ is a polynomial of degree {\it at most} $k$.

\begin{theorem}
  \label{canon_3d} 
  If $u$ and $P\in {\cal E}$ are related by the map (\ref{A}), then, up to an additive constant,
  \begin{equation} \label{R} \int_\infty^{P} W = \log \frac {
      \sigma (u- w_1)\cdots \sigma(u-w_k) }{\sigma (u+w_1)\cdots
      \sigma(u+w_k) } + 2 [ \zeta(w_1) +\cdots+ \zeta(w_k) ]u
    +\varkappa \, u,
  \end{equation}
  where
\begin{equation} \label{res_Q}
\wp(\pm w_i)=z_i, \quad  \frac{ \bar q_k(z_i) }{\bar Q_k'(z_i) }= - \wp'(w_i)\,
\end{equation}
$\zeta(u)$ is the Weierstrass zeta
  function, $\wp'(u)$ is the derivative of the Weierstrass P-function,
and $\varkappa$ is the first coefficient in the expansion
  of $W$ at the infinite point $\infty\in {\cal E}$:
  $W=(\varkappa+ O(u))du$, that is,
$$
\varkappa = \lim_{z \to \infty} \frac{\bar q_k (z) }{\bar Q_k (z) } =b_k.
$$
\end{theorem}

\noindent{\it Proof.} In view of $2\sqrt{ R_3(z_i)}=\wp'(w_i)$,
the condition $\Res\limits_{{\cal P}_i^{-}} \bar W =-1$ is equivalent to \eqref{res_Q}.

It is known (\cite{H_Cu}) that the above integral
has the form
$$
\int_\infty^{P} W = \log \frac { \sigma (u- w_1)\cdots
  \sigma(u-w_k) }{\sigma (u+w_1)\cdots \sigma(u+w_k)} + C_1u +C_0,
\qquad C_1, C_0=\textup{const}.
$$
So, it remains to calculate $C_1$ for the differential \eqref{Omega}.
Differentiate both parts of \eqref{R} by $u$ and evaluate the result
at $u=0$ ($z=\infty$). Then the right hand side gives\footnote{Here we used $\zeta(u)=\sigma'(u)/\sigma(u)$ and oddness of
$\zeta(u)$.}
\begin{gather*}
  \sum_{i=1}^k \left[ \frac{ \sigma'(u-w_i)\sigma (u+w_i)-
      \sigma'(u+w_i)\sigma (u-w_i) }
    {\sigma (u- w_i) \sigma(u+w_i)}+ 2 \zeta(w_i) \right]_{u=0}+ \varkappa \\
  = \sum_{i=1}^k \left[ \zeta(u-w_i)- \zeta(u + w_i)+ 2\zeta(w_i)
  \right]_{u=0} + \varkappa=\varkappa.
\end{gather*}

Derivation of the left hand side of \eqref{canon_3d} gives
$$
\left(\frac{d}{dz} \int_\infty^{P=(z,w)}W
\right)\frac{dz}{du}\bigg |_{u=0}= \lim_{z \to \infty} \frac{\bar q_k
  (z) }{\bar Q_k (z) },
$$
which is precisely $b_k$. $\square$ \medskip

Under a birational transformation $(z,w)\to (\lambda,\mu)$, which
sends $z=\infty$ to $\lambda=c$ and converts $\cal E$ to the even
order curve \eqref{E0},
$$
  \mu^2 = U_4(\la)= -(\lambda -a_{1})(\lambda -a_{2})(\lambda -a_{3})(\lambda -c),
$$
the differential \eqref{Omega} takes the
form
$$
W= \frac{K_k (\la) }{Q_k (\la) }\, \frac{ d \la }{\sqrt{ U_4(\la)} }
$$
with certain degree $k$ polynomials $K_k (\la)$ and
$Q_k(\la)=r_0(\la-r_1)\cdots (\la-r_k)$.  Then Theorem \ref{canon_3d}
implies
\begin{align}
  & \int_{(c,0)}^{P=(\la,\mu)} \frac{ q_k (\la) }{Q_k (\la) }\, \frac{ d \la }{\sqrt{ U_4( \la)}  } = \nonumber \\
  & \quad =\log \frac { \sigma (u- w_1)\cdots \sigma(u-w_k) }{\sigma
    (u+w_1)\cdots \sigma(u+w_k) }
  + 2 [ \zeta(w_1) +\cdots+ \zeta(w_k) ]u + \delta \, u, \\
  w_k & =\int_c^{(r_k,\sqrt{R_4(r_k)})} \frac{ d \la }{ \sqrt{ U_4(
      \la)} }\, d\la , \quad \delta =\frac{K_k(c)}{Q_k (c)}, \label{del}
\end{align}
which is the expression \eqref{del'} in Proposition \ref{3rd_kind}.
Under the birational transformation \eqref{BT}, the condition
\eqref{res_Q} takes the form \eqref{sign_w_j}. $\square$

\paragraph{Proof of Theorem \ref{comp_sol_theorem}}. \\
1). The structure of the solutions \eqref{2nd_kind_sol_} follows from Theorem
\ref{algebraic} and Proposition \ref{3rd_kind}.
Namely, substituting the sigma function expressions \eqref{s_roots}, \eqref{del'}
for each $F_{s,\alpha}(\la)$ and $W_\alpha$ into \eqref{g12}, one obtains these solutions.

Now notice that, in view of the leading behavior \eqref{exp_sig}, the
Laurent expansions of \eqref{2nd_kind_sol_} near $u=0$ and $t=t_0$ are
$$
G_\alpha^{(1,2)} = \frac{\rmi \, \epsilon_\alpha}{ u^k}+ O(u^{-k+1}) =
(-1)^{(k-1)/2}\frac{\epsilon_i}{n^k (t-t_0)^k }+ O((t-t_0)^{-k+1}) .
$$
The leading terms are proportional to those of the required expansions
\eqref{exp_G}, hence the constant factors in the components of
$G_\alpha^{(1,2)}$ are correct.

The structure of $G_\alpha^{(1,2)}(u)$ implies that they are elliptic
functions of the second kind.  Since all the solutions of Poisson
equations are single-valued, Theorem \ref{gen_vector_monodromy} is
fully applicable, hence the vectors $G^{(1,2)}(u)$ must have the vector
monodromy \eqref{vec_momodromy} with certain factors $s_j^2$. Then
$$
G_\alpha^{(1)} (u+ 2\Omega_j) =\pm s_j \, G_\alpha^{(1)} (u), \quad
G_\alpha^{(2)} (u+ 2\Omega_j) =\pm s_j^{-1}\, G_\alpha^{(2)} (u),
$$
On the other hand, from the quasiperiodicity law of $\sigma(u)$ we
have, in particular,
\begin{gather} \label{prod_G_P}
  G_1^{(1)} (u+ 2\Omega_1) =s_1 \, G_1^{(1)} (u), \quad G_1^{(2)} (u+ 2\Omega_1) =s_1^{-1}\, G_1^{(2)} (u), \\
  s_1 = \exp( 2 \Theta_1\Omega_1-2 \Sigma_1\eta_1). \notag
\end{gather}
By the construction of the vectors $G_\alpha^{(1,2)}$ (see
\eqref{det_angular}), for any $u\in {\mathbb C}$
\begin{equation} \label{G_P} G_1^{(1)}(u) P_1(u) + G_2^{(1)}(u) P_2(u)
  + G_3^{(1)}(u) P_3(u) \equiv 0 \, .
\end{equation}
Here, from \eqref{quasi} and \eqref{sols_Pn}, we have
\begin{gather*}
  P(u+2\Omega_1)=(P_1(u), -P_2(u), -P_3(u))^T, \\
  P(u+2\Omega_2)=(-P_1(u), P_2(u), -P_3(u))^T, \\
  P(u+2\Omega_3)=(-P_1(u), -P_2(u), P_3(u))^T .
\end{gather*}
This, together with the integral \eqref{G_P}, and \eqref{prod_G_P}
implies
$$
G^{(1)} (u+ 2\Omega_1) =( s_1 G_1^{(1)} (u), - s_1 G_2^{(1)} (u), -
s_1 G_3^{(1)} (u) )^T.
$$
Repeating the argumentation for other $\Omega_j$, we obtain the
behavior \eqref{mon_alpha}, \eqref{s_j}.  \medskip

3) Let the half-period $\Omega_j$ be imaginary and $\Omega_2$ be real.
Then item (2) of Theorem \ref{gen_vector_monodromy} implies that
$|s_j^2|=1$ if we identify the periods $T_1, T_2$ with some of full
periods $ 4\Omega_1, 4\Omega_2, 4\Omega_3$ of $M_i(u)$.  Hence, also
$|s_j|=1$. By \eqref{s_j}, $s_j=\pm \exp( 2 \Theta_i\Omega_j-2
\Sigma_i \eta_j )$, therefore the argument $2 \Theta_i\Omega_j-2
\Sigma_i \eta_j$ is imaginary. Since $\eta_j$ is imaginary and
$\eta_2$ is real, this means that $2 \Theta_i\Omega_2-2 \Sigma_i
\eta_2$ is real, and $s_2$ is real.

Next, from \eqref{mon_alpha}, we get, for example,
\begin{equation} \label{quo_G} \frac {G_1^{(1)}}{G_3^{(1)}}
  (u+2\Omega_1)= - \frac {G_1^{(1)}}{G_3^{(1)}} (u), \quad \frac
  {G_1^{(1)}}{G_3^{(1)}} (u+2\Omega_2)= \frac {G_1^{(1)}}{G_3^{(1)}}
  (u) .
\end{equation}
Hence, $G_1^{(1)}(u)/G_3^{(1)}(u)$ is an elliptic function with the
periods $ 4\Omega_1, 2\Omega_2$.  Its zeros and poles in the
parallelogram of periods are
\begin{gather*}
  \{w_{1,1},\dots, w_{k,1}, w_{1,1}+2\Omega_1,\dots,
  w_{k,1}+2\Omega_1\} \quad
  \text{and, respectively,} \\
  \{w_{1,3},\dots, w_{k,3}, w_{1,3}+2\Omega_1,\dots,
  w_{k,3}+2\Omega_1\}.
\end{gather*}
Then, according to the Abel theorem, the difference of their sums must
be zero modulo the lattice $\{4\Omega_1{\mathbb Z} + 2\Omega_2
{\mathbb Z}\}$:
$$
2\Sigma_1 -2 \Sigma_3 \in \{ 4\Omega_1{\mathbb Z} + 2\Omega_2 {\mathbb
  Z} \},
$$
therefore, $\Sigma_1 -\Sigma_3 \in \{ 2\Omega_1{\mathbb Z} + \Omega_2
{\mathbb Z}\}$. Further, the case $\Sigma_1 -\Sigma_3 \in \{
2\Omega_1{\mathbb Z} + 2\Omega_2 {\mathbb Z}\}$ is not possible
(otherwise $G_1^{(1)}(u)/G_3^{(1)}(u)$ would had been a product of an
elliptic function with the period lattice $\{2\Omega_1{\mathbb Z} +
2\Omega_2 {\mathbb Z}\}$ and an exponent, i.e., not doubly-periodic
itself). Hence,
$$
\Sigma_1 -\Sigma_3 \equiv \Omega_2 \quad \text{mod} \quad
\{2\Omega_1{\mathbb Z} + 2 \Omega_2 {\mathbb Z} \}.
$$
Applying the same argumentation to the quotients
$G_2^{(1)}(u)/G_3^{(1)}(u)$, $G_1^{(1)}(u)/G_2^{(1)}(u)$, we arrive at
the first part of relations \eqref{diff_Sigma_Theta}.

Next, the quasi-periodicity of $\sigma(u)$ implies
\begin{gather*}
  \frac {G_1^{(1)}}{G_3^{(1)}} (u+2\Omega_1)=
  \exp\left[2(\Sigma_1 -\Sigma_3)\Omega_1-2(\Theta_1 -\Theta_3)\eta_1 \right]\frac {G_1^{(1)}}{G_3^{(1)}} (u), \\
  \frac {G_1^{(1)}}{G_3^{(1)}} (u+2\Omega_2)= \exp\left[
    2(\Sigma_1-\Sigma_3) \Omega_2-2(\Theta_1 -\Theta_3)\eta_2
  \right]\frac {G_1^{(1)}}{G_3^{(1)}} (u) .
\end{gather*}
Comparing this with \eqref{quo_G} and using $\Sigma_1 -\Sigma_3 \equiv
\Omega_2$, as well as the known Legendre relations
$$
\eta_2 \Omega_3 - \eta_3 \Omega_2 = \eta_3
\Omega_1-\eta_1\Omega_3=\eta_1\Omega_2-\eta_2\Omega_1=\frac 12 \pi \,
\rmi \, ,
$$
we get the second half of \eqref{diff_Sigma_Theta}.

4). Since $G^{(1)}(u)$ is a solution of the Poisson equations, it must
satisfy the integral
$$
{ G_1^{(1)}}^2(u) + {G_2^{(1)}}^2(u)+ {G_3^{(1)}}^2(u)= M.
$$
Then
\begin{gather*}
  {G_1^{(1)}}^2(u+2n\Omega_2) + {G_2^{(1)}}^2(u+2n\Omega_2)+ {G_3^{(1)}}^2(u+2n\Omega_2)= \qquad \\
  \qquad =s_2^{2n} \left[ { G_1^{(1)}}^2(u) + {G_2^{(1)}}^2(u)+
    {G_3^{(1)}}^2(u) \right]=M \qquad \forall \; n\in {\mathbb Z}.
\end{gather*}
Since $|s_2|\ne 1$, letting above $n\to \infty$ or $n\to -\infty$, we
conclude that $M$ must be zero. The same argumentation applied to
$G^{(2)}(u)$, give the second identity in \eqref{zero_sum}.

Let now $\Omega$ be a half-period of $E$, $\eta=\zeta(\Omega)$ and $w
\in {\mathbb C}$ an arbitrary number. From the identity
$$
\frac{ \sigma ( \Omega +w ) \, \sigma(\Omega-w)}{\sigma^2 (\Omega)
  \,\sigma^2(w)} =\wp(\Omega)- \wp(w),
$$
using the quasiperiodicity of $\sigma(u)$, we get
\begin{equation} \label{symm_sigma} - e^{2\eta \, w} \frac{
    \sigma^2(\Omega -w)}{\sigma^2(\Omega) \,\sigma^2(w)} =
  \wp(\Omega)- \wp(w).
\end{equation}
Applying this to the solutions \eqref{2nd_kind_sol_} in the case
$\Omega=\Omega_i, w=w_{l,i}$ and using \eqref{s_j} gives us
\begin{align*}
  \left[ G_\alpha^{(1)} (\Omega_j)\right]^2 & = \epsilon_\alpha^2
  \prod_{l=1}^k e^{2\zeta(w_{\alpha,l}) \Omega_j} \frac{
    \sigma^2(\Omega_j -w_{\alpha,l})}{\sigma^2(\Omega_j)
    \,\sigma^2(w_{\alpha,l})} = (-1)^{\delta_{\alpha
      j}}\epsilon_\alpha^2 \, s_j e^{2\Sigma_\alpha \eta_j}
  \prod_{l=1}^k
  \frac{ \sigma^2(\Omega_j -w_{\alpha,l})}{\sigma^2(\Omega_j) \,\sigma^2(w_{\alpha, l})} \\
  & = (-1)^{\delta_{\alpha j}} \epsilon_\alpha^2 s_j \prod_{l=1}^k
  e^{2w_{\alpha,l}\, \eta_j}
  \frac{ \sigma^2(\Omega_j -w_{\alpha,l})}{\sigma^2(\Omega_j) \,\sigma^2(w_{\alpha,l})} \\
  & = - (-1)^{\delta_{\alpha j}} \epsilon_i^2 \, s_j \cdot
  (\wp(\Omega_j)-\wp(w_{\alpha,1})) \cdots (\wp(\Omega_j)-
  \wp(w_{\alpha,k})),
\end{align*}
i.e., the expressions \eqref{G^2} for $G_i^{(1)}(\Omega_j)$. The expressions for $G_i^{(2)}(\Omega_j)$
are obtained in the same way.

Next, the structure of the solutions \eqref{2nd_kind_sol_} gives
\begin{equation} \label{sums_prods_G}
\sum_{\alpha=1}^3 G^{(1)}_\alpha (u) G^{(2)}_\alpha (u) =
- \sum_{\alpha=1}^3 \epsilon_\alpha^2 \prod_{l=1}^k (\wp(u)-\wp(w_{l,\alpha})).
\end{equation}
On the other hand, as follows from \eqref{zero_sum} and \eqref{G^2}, for any $\alpha=1,2,3$
\begin{gather*}
\epsilon_\alpha^2 \prod_{l=1}^k (\wp(\Omega_\alpha)-\wp(w_{l,\alpha})) = \epsilon_\beta^2
\prod_{l=1}^k (\wp(\Omega_\alpha)-\wp(w_{l,\beta})) + \epsilon_\gamma^2
\prod_{l=1}^k (\wp(\Omega_\alpha)-\wp(w_{l,\gamma})), \\
(\alpha,\beta,\gamma)=(1,2,3).
\end{gather*}
Then, for $u=\Omega_j$, $j=1,2,3$, relation \eqref{sums_prods_G} yields
$$
- \sum_{\alpha=1}^3 G^{(1)}_\alpha (\Omega_j)\, G^{(2)}_\alpha (\Omega_j)
=2 \epsilon_j^2 \prod_{l=1}^k (\wp(\Omega_j)-\wp(w_{l,j})) \, .
$$
Since $\langle G^{(1)}(u), G^{(2)}(u)\rangle$ is a first integral, we get the relations \eqref{eq_Pi}.

It remains to prove that the latter equal $\Pi$. According to Theorem \ref{algebraic} and relation \eqref{QQ},
$$
Q_{k,\alpha}(\la)=r_{0,\alpha} \prod_{l=1}^k(\la-r_{i,\alpha})=\text{const}
(P_\beta^2(\la)+P_\gamma^2(\la))\cdot (\la-c)^k
$$
for $(\alpha,\beta,\gamma)=(1,2,3)$. This implies, in particular,
\begin{equation} \label{aux2}
C \,(\wp(u)-\wp(w_{1,1})) \cdots (\wp(u)-\wp(w_{1,k})) = P_2^2(u)+ P_3^2 (u)
\end{equation}
with a certain constant $C$. On the other hand,
in view of the solutions \eqref{P^2}, one obtains,
\begin{align*}
 P_2^2(u)+ P_3^2 (u) & =
\epsilon_2^2 \left(\wp(u)-\wp(\Omega_2) \right) \prod_{l=1}^s \left(\wp(u)-\wp(v_{l,2})\right)^2  \\
& \quad + \epsilon_3^2 \left(\wp(u)-\wp(\Omega_3) \right) \prod_{l=1}^s \left(\wp(u)-\wp(v_{l,3})\right)^2 .
\end{align*}
Letting in the above two relations $u\to \infty$, taking into account the expansion $\wp(u)=1/u^2 + O(1)$,
and comparing the leading coefficients of $1/u^{2k}$, we find
$$
C= \epsilon_2^2 + \epsilon_3^2 = \frac{1}{{(a_2-a_1)(a_2-a_3)}} +  \frac{1}{{(a_3-a_1)(a_3-a_1)}}= -\epsilon_1^2.
$$
Since $P_1(u=\Omega_1)=0$, we have
$\Pi=P_1^2(u)+ P_2^2(u)+ P_3^2(u)= P_2^2(\Omega_1)+ P_3^2(\Omega_1)$. Comparing this with \eqref{aux2} for $u=\Omega_1$,
we get
$$
\Pi =- \epsilon_1^2 \prod_{l=1}^k (\wp(\Omega_1)-\wp(w_{l,1})) \, ,
$$
and, therefore, the last equality in \eqref{eq_Pi}. $\square$

\end{document}